\documentclass[aps,prd,amsfonts,floats,floatfix,eqsecnum,notitlepage,nofootinbib,superscriptaddress]{revtex4-2}

\usepackage{mathrsfs}
\usepackage{amsfonts}
\usepackage{amsmath}
\usepackage{amsthm,amssymb}
\usepackage[normalem]{ulem}

\usepackage{graphicx}
\usepackage{graphics}
\usepackage{color}

\usepackage{tikz}
\usepackage{subcaption}
\usepackage{amsfonts}
\usetikzlibrary{decorations.pathmorphing,patterns,calc}
\usepackage{bm}
\usepackage{pifont}
\usepackage{pgfplots}
\usepackage[absolute,overlay]{textpos}
\usepackage{MnSymbol}
\usetikzlibrary{arrows.meta}

\theoremstyle{plain}
\newtheorem{theorem}{Theorem}[section]
\newtheorem{proposition}{Proposition}[section]
\newtheorem{lemma}{Lemma}[section]

\theoremstyle{definition}
\newtheorem{comments}{Comment}
\newtheorem{definition}{Definition}[section]

\newcommand{\be}{\begin{equation}}
\newcommand{\ee}{\end{equation}}

\DeclareMathOperator{\arccosh}{arccosh}

\newcommand{\real}{\mathbb{R}}

\newcommand{\pd}[2]{\frac{\partial #1}{\partial #2}}

\newcommand{\sigmcv}{\Sigma_{McV}}
\newcommand{\mmcv}{{\cal{M}}_{McV}}
\newcommand{\sigrw}{\Sigma_{RW}}
\newcommand{\mrw}{{\cal{M}}_{RW}}
\newcommand{\mm}{{\cal{M}}}
\newcommand{\sig}{\Sigma}
\newcommand{\tsig}{t_\sig}
\newcommand{\rsig}{r_\sig}
\newcommand{\rosig}{\rho_\sig}

\newcommand{\tasig}{\tau_\sig}

\newcommand{\choice}{{1,2}}
\newcommand{\del}{\Delta}
\newcommand{\delmcv}{\Delta_{McV}}
\newcommand{\delrw}{\Delta_{RW}}

\newcommand{\cj}{{\cal{J}}}

\newcommand{\hor}{{\cal{H}}}
\newcommand{\jd}{\dot{J}}
\newcommand{\sign}[1]{\textrm{sign}({#1})}

\pgfplotsset{compat=1.16}

\begin{document}
    
\author{Brien C. Nolan}
\email{brien.nolan@dcu.ie}
\affiliation{Centre for Astrophysics and Relativity (CfAR), School of Mathematical Sciences, Dublin City
University, Glasnevin, Dublin, D09 V209, Ireland.}
\affiliation{School of Theoretical Physics, Dublin Istitute for Advanced Studies, 10 Burlington Road, Dublin D04 C932, Ireland}

\title{On a cosmological Oppenheimer-Snyder model: matching McVittie and FLRW spacetimes.}

\begin{abstract}
We consider the necessary and sufficient conditions for the smooth matching of an expanding McVittie spacetime and a spatially flat Friedmann-Lema\^{\i}tre-Robertson-Walker (FLRW) spacetime across a general hypersurface - that is, a hypersurface of arbitrary causal character, where the character possibly changes from point to point. We identify a number of special cases, {giving a complete description of (for example) constant radius matching hypersurfaces.} In the general case, we show that the full set of matching conditions is equivalent to a 3-dimensional non-linear system of first order ODEs, subject to a single initial value constraint. Taking the perspective that the McVittie region is specified, we prove a semi-global existence result for this system and interpret these solutions in terms of the physical set-up. Solutions exist and are unique globally to the future. In this direction, the required FLRW spacetime exists and is unique and the matching hypersurface exists almost everywhere and is unique. However, the necessary embedding conditions for the hypersurface necessarily fail at least once - but only at isolated points - along the putative hypersurface. To the past, the matching hypersurface extends to meet the past singularity of the McVittie region. We determine the causal character of the matching hypersurface in these limits, and consider the implications for the existence (or otherwise) of an isotropic source of the McVittie cosmological black hole spacetime. We find that generically, the matching hypersurface has spacelike as well as timelike portions, and so the isotropic region cannot be considered to be an interior for the McVittie exterior. This is in distinction to the Oppenheimer-Snyder model, in which an isotropic dust sphere matches across an everywhere-timelike boundary to the Schwarzschild exterior. On the basis of the sometimes-spacelike nature of the hypersurface, we conclude that an isotropic source of the McVittie spacetime does not exist in a global sense. 
\end{abstract}

\maketitle
\tableofcontents

{\section{Introduction and summary}\label{sec:intro}}

McVittie's spacetime \cite{mcvittie1933mass} provides an exact model of a black hole embedded in an isotropic cosmological background (i.e.\ a Friedmann-Lema\^{\i}tre-Robertson-Walker (FLRW) spacetime). The precise details of the meaning of this statement have been developed over the last 25 years. Global features of the spacetime were first considered in \cite{sussman1988spherically}, and were further developed in \cite{nolan1999point}. A crucial advance was made in \cite{kaloper2010mcvittie}, in which the authors demonstrated that in the expanding case, there is a future boundary of the past of future null infinity that lies at finite affine distance along an open set of ingoing radial null geodesics. This corresponds to the existence of a black hole horizon in the spacetime, in line with the spirit of the classical definition found in, for example, \cite{hawking1973large}. Extending across this boundary leads one to a region of Schwarzschild-de Sitter spacetime, and in fact this may be a black or white hole region of that spacetime \cite{kaloper2010mcvittie,lake2011more} (and see \cite{da2013expansion} for a characterization of when McVittie spacetimes admit this white hole extension). More recently, another key feature of the spacetime was established that further demonstrates its black hole credentials: test particles can free-fall to the horizon in finite proper time \cite{nolan2025can}. 

However, there are reasons to be cautious with this  interpretation. Principal among these is the fact that the spacetime ({each member of the family}) admits a spacelike curvature singularity at area radius $r=2m$, where $m>0$ is the constant mass parameter that plays the role of the Schwarzschild mass in the spacetime. Thus McVittie spacetimes do not have a zero-volume big bang singularity. In the expanding case, this spacelike singularity forms a past boundary of the spacetime: all surfaces of constant cosmological time intersect the singularity and all future-directed causal geodesics meet the singularity at finite proper time/affine distance in the past. Thus there appears to be a significant obstruction to obtaining McVittie's metric as the solution of an initial value problem with data posed on a regular hypersurface, and the black hole horizon does not appear to be evolutionary in nature: it does not arise from a regular initial configuration - one without singularities or trapped surfaces. {This contrasts with} the asymptotically flat case, it is reasonably straightforward to obtain a spherically symmetric black hole from regular initial data: this can be achieved with the Oppenheimer-Snyder model \cite{oppenheimer1939continued} describing the gravitational collapse of a homogeneous dust sphere.

The aim of this paper is to determine if there is a cosmological Oppenheimer-Snyder (OS) model: an FLRW region matched across a smooth boundary to a McVittie spacetime. To consider the FLRW region as a source for the McVittie spacetime, we should require that the boundary hypersurface $\Sigma$ is timelike - as it is in the OS case. Then we can view the FLRW region as an interior region, spatially bounded at every instant of time and with a boundary that an observer could cross into (or from) the McVittie region which has infinite extent. However, we will not impose the restriction that $\Sigma$ is timelike, firstly for reasons of generality, and secondly because this feature should emerge from the matching conditions, as it does in the OS model. Assuming a timelike boundary may be too restrictive in another sense: the matching conditions may break down not for any fundamental reason, but precisely because the matching hypersurface has been assumed to be timelike, and the equations cannot accommodate the hypersurface becoming null. So this is the problem that we consider here: can we achieve the smooth matching of a given McVittie spacetime with an FLRW spacetime? If so, what are the global features (including causal character) of the matching hypersurface? Does the matched configuration admit the interpretation of a cosmological OS model - collapse from a regular initial configuration to a McVittie black hole? 

Our aim here is to be comprehensive: to give a more or less complete account of the matching of (spatially flat) McVittie and FLRW spacetimes. We will consider general matching hypersurfaces $\Sigma$, meaning that the causal character is not fixed, and may vary from point to point \cite{mars1993geometry}. Our aim of being as complete as possible necessarily means dealing with a number of different cases (imposing some, all or none of the conditions that the matching hypersurface has open subsets on which it is null, has constant radius or is foliated by marginally trapped surfaces) that require a separate treatment. 
As an aid to reading, we provide an extended summary below. First, we review related work. 

Matching of McVittie spacetimes has been considered previously on only a small number of occasions. In \cite{nolan1993sources}, the present author derived a source for McVittie spacetime in the form of a spherically symmetric, shear free, spatially homogeneous and conformally flat solution of the Einstein equations with perfect fluid source. This was matched across a co-moving (and hence timelike) boundary to a McVittie exterior. However, this model has pressure singularities: the boundary must meet the past singularity of the McVittie region at a finite time in the past. Since the pressure of the interior matches that of the exterior (McVittie) region \cite{nolan1993sources}, it must diverge at this time - the McVittie pressure diverges at $\{r=2m\}$ (see (\ref{mcv-matter}) below). Even before reaching this point in the past, the pressure of the interior may diverge at some co-moving radius less than that of the co-moving boundary radius. This occurs at sufficiently early times, but times later than that at which the boundary meets $\{r=2m\}$. The pressure singularity is absent if the scale factor $a(t)$ is bounded below by the ratio of the mass parameter $m$ and co-moving radius of the boundary, which is implicitly assumed in \cite{nolan1993sources}.

In \cite{nandra2013dynamics}, Nandra \textit{et al.} considered further aspects of this and related models in their broader study of the dynamics of a uniform density sphere - the interior, with homogeneous density $\rho_\textrm{i}(t)$ and Hubble function $H_\textrm{i}(t)$ - embedded in a background that also has uniform (but possibly distinct) density - the exterior, with $\rho_\textrm{e}(t), H_\textrm{e}(t)$. The regions are matched across a timelike hypersurface $r=a(t)$, where $r$ is the area radius. The model entails the mass relation 
\renewcommand\arraystretch{1.5}
\be E(t,r) = \left\{ \begin{array}{ll} \frac43 \pi\rho_\textrm{i}(t), & r\leq a(t); \\
\frac43 \pi\rho_\textrm{e}(t) + m(t), & r> a(t),
\end{array} \right. \label{Nandra-mass} \ee
where $E$ is the Misner-Sharp mass (see (\ref{ems-def}) below), and continuity of $E$ across the boundary is assumed. Remarkably, specifying that $H_\textrm{i}=H_\textrm{e}$ leads to $m=$ constant and a McVittie exterior, and to the homogeneous interior considered in \cite{nolan1993sources}. The pressure singularity discussed above is identified, and the authors seek to avoid this by taking the boundary radius to be sufficiently large. However their equations of motion for the boundary radius indicate that the necessary lower bound must be exceeded, and so the pressure singularities are also a feature of this model. The authors discuss additional features of their model, including the expanding nature of the boundary, and its reduction to the standard constant density Schwarzschild interior in the (exterior) vacuum limit. The use of area-radius coordinates greatly simplifies the form and interpretation of the interior metric. 

In \cite{haines1993thin}, the authors applied Israel's thin shell formalism \cite{israel1966singular} to consider a variety of configurations in which a Schwarzschild or FLRW interior is joined across a spherical thin shell - a timelike hypersurface with non-vanishing surface energy-stress-momentum tensor - to a McVittie exterior. The authors investigated numerically the influence of the McVittie mass parameter $m$ on the motion of the shell, finding that a larger value of $m$ leads to slower expansion and a greater propensity to collapse. A negative value of the mass parameter increases the expansion rate.  

Thin shell boundaries were also considered in \cite{tang2025matching}. Here, the authors considered the case where two McVittie spacetimes are joined across a spherical thin shell (this includes the special cases where one or other or both correspond to Minkowski spacetime, Schwarzschild spacetime or an FLRW spacetime). The analysis includes the case where an FLRW interior is joined across a thin shell to a McVittie exterior. The results reinforce and add detail to those of \cite{haines1993thin}: an analysis of the role of the initial velocity of the shell is included. 

We note that these two studies are quite different from the present one. Extending the seminal work of Israel \cite{israel1966singular}, Lake derived an equation of motion for a timelike spherical shell joining two spherically symmetric spacetimes \cite{lake1984equation}. The preliminary matching conditions (continuity of the metric across the shell) are assumed to hold. The system closes by assuming an equation of state for the matter content of the shell: \cite{haines1993thin} considers the cases of dust $P_\textrm{shell}=0$ and false vacuum $\sigma_\textrm{shell}+P_\textrm{shell}=0$ and \cite{tang2025matching} considers only dust - although energy conditions on the shell are not discussed. Here, $\sigma_\textrm{shell}$ and $P_\textrm{shell}$ are the shell's surface energy density and pressure respectively. In the present case, the equations of motion come from the matching conditions which incorporate the fact that there is \textit{no} thin shell at the boundary. Lake's equation of motion is singular in the `no shell' limit, and it is not clear how to connect the two scenarios. 

{We now summarise the remainder of the paper.}


{\subsection{McVittie and FLRW geometry}\label{subsec:McV-FLRW}}

In Section \ref{sec:McV-FLRW}, we review the metric and matter of McVittie and FLRW spacetimes. The McVittie spacetime is specified by a choice of positive constant $m$ (the Schwarzschild mass parameter) and Hubble function  $H(t)$, and we use coordinates corresponding to cosmic time $t$ and the area radius $r$. (Recall that $H(t)=a'(t)/a(t)$ where $a$ is the scale factor.) We can consider the FLRW spacetime to be a special case of McVittie with $m=0$. We use coordinates $(\tau,\rho)$ corresponding to the cosmic time and area radius of the FLRW spacetime, and we write the Hubble function  as $\cj(\tau)$. We also mention relevant global properties of McVittie spacetimes. We assume that both $t\mapsto H(t)$ and $\tau\mapsto \cj(\tau)$ are $C^2$, but {for now}, we do not specify any further conditions. Such conditions (sign, monotonicty, past and future asymptotics) will ultimately play an important role in our main result (Theorem \ref{Thm:Match-V4}), and are discussed in detail at that point. 

{\subsection{The matching hypersurface and the preliminary matching conditions}\label{subsec:matching-prelim}}

In Section \ref{sec:matching-prelim} we introduce a general hypersurface $\Sigma^+$ of McVittie spacetime (which readily specializes to a general hypersurface $\Sigma^-$ of FLRW spacetime). We describe the preliminary matching conditions, introducing the necessary vector fields on the tangent space of the matching hypersurface. We treat  $\Sigma\equiv\Sigma^\pm$ as embeddings $\Phi^\pm:\sigma\to\mathcal{M}^\pm$ of an abstract 3-manifold $\sigma$ into the respective spacetime manifolds $\mathcal{M}^\pm$. In local coordinates $\xi^a$ on $\sigma$ and $x_\pm^\alpha$ on $\mathcal{M}$, we write the embeddings as 
\be \Phi^\pm:\xi^a\mapsto x^\alpha_\pm = x^\alpha_\pm(\xi^a).\label{embedpm} \ee We can then write down the \textit{preliminary matching conditions}. These entail equality of the pull-back (under $\Phi^\pm$) to $\sigma$ of the metrics $g^{\pm}$ induced on $\Sigma^\pm$. We use local coordinates $\xi^a=(\lambda,\theta,\phi)$ on $\sigma$. The preliminary matching conditions allow us to identify the spherical coordinates $(\theta,\phi)$ on $\sigma$ with the corresponding coordinates $(\theta^\pm,\phi^\pm)$ on $\mathcal{M}^\pm$: these play no further role. Likewise, we find that the area radius is continuous across $\Sigma^\pm$ (and is as smooth on $\sigma$ as the assumed embeddings), and so we have a single area radius $r\equiv\rho$ throughout the matched spacetimes. The tangent basis of $\mathcal{M}^\pm$ at $\Sigma^\pm$ comprises 
\be \vec{e}_a^\pm = (e^\alpha_a)^\pm\left.\pd{}{x^\alpha}\right|_{\Sigma^\pm},\quad (e^\alpha_a)^\pm=\pd{x_\pm^\alpha}{\xi^a} \label{tangents} \ee the push-forward by $\Phi^\pm$ of the tangents $\partial_{\xi^a}$ of $\sigma$  to $\Sigma^\pm$, along with  rigging vectors $\vec{\ell}^\pm$ on $\Sigma^\pm$. These are vectors that are everywhere transverse to $\Sigma^\pm$. We impose the \textit{rigging compatibility conditions} 
\be g^+(\vec{\ell}^+,\vec{\ell}^+) = g^-(\vec{\ell}^-,\vec{\ell}^-),\quad \ell^+_a = \ell^-_a,\label{rig-compat} \ee
where $\ell_a^\pm = (e^\alpha_b\ell_\alpha)^\pm$.

{\subsection{Junction conditions}\label{subsec:junction-conditions}}
In Section \ref{sec:junction-conditions}, we describe the full set of junction conditions that are the core concern of this paper. These arise from the requirement that there is no delta-function singularity in the Riemann tensor of the matched spacetime \cite{mars1993geometry}. We are then able to address the main purpose of this paper: to describe all cases where (expanding) McVittie spacetimes can be matched across a (general) hypersurface to an FLRW spacetime. Based on these conditions, we identify a number of subcases that require a separate treatment. We need to distinguish between cases where $\Sigma\equiv\Sigma^\pm$ is null on an open set, or only (at most) at isolated points, and between cases where the area radius is constant on (an open subset of) $\Sigma$. In each case, we identify the relevant full set of matching and junction conditions, using the general hypersurface formalism of \cite{mars1993geometry}. {We complete the analysis in the case of matching across a null hypersurface of constant radius. We find that in this case, the only possible matching is when both the McVittie and FLRW regions degenerate to regions of (the same) de Sitter spacetime. We defer the problem of solving the matching conditions in the other cases to Section \ref{sec:solving}.} 

{\subsection{Consequences of matching}\label{subsec:consequences}}
In Section \ref{sec:consequences}, we consider various physical and geometric consequences of matching, and consider some special cases. In Section \ref{subsec:conseq-mass}, we consider the implications of matching for the Misner-Sharp mass of the matched spacetimes: this leads to the interpretation of McVittie spacetimes as either the exterior of an over-dense isotropic interior region, or the under-dense local interior of an isotropic exterior. In both cases, the mass parameter $m$ accounts for the `missing' matter density. 
In Section \ref{subsec:conseq-trapping}, we deal with questions of horizons and trapping, which provides a useful guide when it comes to solving the matching conditions. As shown in \cite{fayos1996general}, trapped regions of one spacetime can match only with trapped regions of the other (and likewise for anti-trapped regions). We revisit this result, and tie it to the question of time-orientability of the matched spacetime.  
Since Schwarzschild(-de Sitter) spacetime is a special case of a McVittie spacetime, our analysis should accommodate the description of the Oppenheimer-Snyder model of gravitational collapse \cite{oppenheimer1939continued} (and equivalently, the Einstein-Straus vacuole model \cite{einstein1945influence}). In Section \ref{subsec:conseq-OS}, we show how the discussion can be specialized to this case.  
In Section \ref{subsec:conseq-matter}, we consider the Israel conditions and their implications for the matter (dis)continuities at the matching hypersurface. We consider the special case of matching across a co-moving timelike hypersurface, and obtain {an essentially complete description of the possible matching configurations in the case when the matching hypersurface is co-moving with the cosmological fluid on one side or the other of the boundary, or with both. We emphasize that this is a highly specialized situation, possible only for a very limited subclass of McVittie spacetimes.} 

{\subsection{Solving the matching conditions}\label{subsec:solving}}
Section \ref{sec:solving} contains the main results of this paper. There are four cases to consider in all: three cases carry forward from Section \ref{sec:junction-conditions}, and one of these presents two subcases. We can treat three of these four cases relatively quickly (Section \ref{subsec:sol-null-r-noncon} - null matching hypersurfaces of non-constant radius; Section \ref{subsec:sol-non-null-r-con} - non-null matching hypersurfaces of constant radius; Section \ref{subsubsec:hor} - non-null hypersurfaces of non-constant radius foliated by marginally trapped surfaces). This leaves what should be considered the general case of a non-null hypersurface of non-constant radius that is not foliated by marginally trapped surfaces. This is the subject of the remainder of the paper. Thus Section \ref{subsubsec:general} provides a discussion of existence and uniqueness of matching configurations (once the expanding McVittie spacetime has been specified) in this most general case. The essence of our main result (Theorem \ref{Thm:Match-V4}) is that once a given McVittie spacetime is specified (so that the parameter $m$ and Hubble function $H(t)$ are given), and an initial point of a matching hypersurface is chosen along with an initial value $\mathcal{J}_0$ of the FLRW Hubble function, then there is a unique FLRW metric and a unique matching hypersurface $\Sigma$ such that the full set of matching and junction conditions are satisfied. The FLRW spacetime must also be expanding everywhere: matching is not possible for $\mathcal{J}_0<0$. We also provide results describing various global features of the matching hypersurface that are a consequence of the matching. The matching hypersurface exists globally in the future direction, but must originate at a finite time in the past at the past singularity $\{r=2m\}$ of the McVittie spacetime. Thus the matching hypersurface $\Sigma$ cannot fully excise this singularity from the matched spacetime. At early times ($t\gtrsim t_1$, where $t_1$ is the time at which the matching hypersurface reaches $\{r=2m\}$), $\tau$ - the FLRW cosmic time coordinate - is \textit{decreasing} with increasing $t$, the McVittie time coordinate. Related to this is the finding that the matching hypersurface must be initially spacelike. We also describe how the late-time causal character of $\Sigma$ depends on the background equation of state parameters of the McVittie region. 

{\subsection{Conclusions}\label{subsec:conclusions}}
In Section \ref{sec:conclusions}, we discuss the implications of our results (mainly those of Section \ref{sec:solving}) for the central question posed above: can we construct an isotropic source for a given McVittie spacetime?   

{\subsection{Notation}\label{subsec:notation}} Many (probably most) of the equations below hold only on the matching hypersurface $\Sigma$. To avoid over-complicating things, we will not generally  use any special notation to indicate equations on $\Sigma$ (e.g.\ $\left.Q\right|_\Sigma=0$ or $Q\stackrel{\Sigma}{=}0$) - but this will be done when the emphasis seems necessary. The context should make clear when a given equation holds only on $\Sigma$, and when it holds elsewhere. We use the common notation $[X]=X^+-X^-$ for the difference of a tensor quantity evaluated on $\Sigma$ using the embeddings on either side. We use units with $G=c=1$, and follow the curvature conventions of \cite{wald1984general}.

\section{McVittie and FLRW geometry}\label{sec:McV-FLRW}

\subsection{McVittie}\label{subsec:mcvittie}

In the McVittie region, we use coordinates $x^\alpha=(t,r,\theta,\phi)$, where $r$ is the area radius of the spherically symmetric spacetime, and $t$ is a cosmic time coordinate, with the property that the future-directed, unit fluid flow covector is parallel to $dt$. {The time-orientation of the spacetime comes from $t$, which is assumed to increase into the future.} The coordinates $(\theta, \phi)$ are the usual coordinates on the unit 2-sphere. Then the McVittie metric is uniquely determined by a (positive) constant $m$ - the mass parameter - and a $C^2$ function $H(t)$, which corresponds to the Hubble function of the background FLRW spacetime in which the mass is embedded. {The details below can be found in \cite{nolan2014particle}.} In these coordinates, the line element has the form
\begin{equation} ds_{McV}^2 = -\alpha dt^2-2\beta dtdr+\gamma dr^2 + r^2d\Omega^2,\label{mcv-lel}\end{equation}
where 
\begin{eqnarray}
\alpha &=& 1-\frac{2m}{r}-r^2H^2,\label{alpha-def} \\
\beta &=& r H\gamma^{1/2}, \label{beta-def}\\
\gamma &=& f^{-1},\quad f= \left(1-\frac{2m}{r}\right)\label{gam-def}
\end{eqnarray}
and $d\Omega^2=d\theta^2+\sin^2\theta d\phi^2$ is the line element of the unit 2-sphere. 

Note that 
\begin{equation} \alpha\gamma+\beta^2 = 1,\label{al-be-ga-id}\end{equation} and that the time coordinate $t$ is invariantly defined up to translation. 
The metric provides a solution of Einstein's equation with a perfect fluid source, {where the fluid flow covector is parallel to ${dt}$. Then the (unit, future-pointing) fluid velocity vector field is given by 
\begin{equation} \vec{u}_{McV} = \gamma^{1/2}\pd{}{t}+\beta \gamma^{-1/2}\pd{}{r}.\label{u-mcv}\end{equation}} The energy density and pressure {are given by (\cite{nolan2014particle})}
\begin{equation} 8\pi\mu_{McV}(t,r) = 3H^2(t)-\Lambda,\qquad 8\pi P_{McV}(t,r) = -2H'(t)\gamma^{1/2}-3H^2(t)+\Lambda, \label{mcv-matter}\end{equation}
where $\Lambda$ is the cosmological constant. Here and throughout, a prime ($'$) represents the derivative with respect to argument. {The weak energy condition is equivalent to 
\be |H(t)|\geq \sqrt{\frac{\Lambda}{3}}\quad \hbox{and}\quad H'(t) \leq 0. \label{eq:wec}\ee}
The expansion of the fluid flow lines is given by $\theta=3H$, and so the sign of $H$ distinguishes between a collapsing ($H<0$) and an expanding ($H>0$) spacetime (or region of spacetime). There is a curvature singularity at $r=2m$ except in the special case of vanishing $H'(t)$, and the line element is defined only for $r>2m$. We assume throughout that $m>0$\footnote{It seems that the possible global structures of McVittie spacetimes with $m<0$ have not been considered in the literature. As mentioned above, a negative value influences the shell dynamics considered in \cite{haines1993thin} and \cite{tang2025matching}.}.
We recall the following features \cite{mcvittie1933mass,nolan1998point}:
\begin{enumerate}
\item With $m=0$, equation (\ref{mcv-lel}) gives the line element of a spatially flat {FLRW} universe with Hubble function $H(t)$;
\item With $H(t)=0$, the line element corresponds to the exterior Schwarzschild spacetime and with $H(t)=H_0=$ constant, the line element corresponds to Schwarzschild-de Sitter spacetime with mass parameter $m$ and cosmological constant $\Lambda = 3H_0^2$;
\item The singularity at $r=2m$ is spacelike, and in the expanding case, forms a past boundary of the spacetime \cite{nolan1999point};
\item In the expanding case and for $\Lambda\geq 0$, the spacetime contains a black hole horizon \cite{kaloper2010mcvittie,lake2011more,nolan2017local}.
 
\end{enumerate}

\subsection{FLRW}\label{subsec:flrw}

It will be convenient also to express the Robertson-Walker region in area-radius coordinates rather than the usual co-moving coordinates. 
We anticipate a subset of the preliminary matching conditions (continuity of the metric tensor across the matching hypersurface) by using the same labels for the coordinates on the 2-sphere in the FLRW region as we used in the McVittie region. We emphasize that this can be done without loss of generality. So in the FLRW region, we write
\begin{equation} ds_{RW}^2 = - (1-\rho^2\cj^2(\tau))d\tau^2 - 2\rho\cj(\tau)d\rho d\tau + d\rho^2 + \rho^2d\Omega^2, \label{rw-lel}\end{equation} 
where $\tau$ is a cosmic time coordinate (proper time along the geodesic fluid flow lines) and $\cj(\tau)$ is the Hubble function of the FLRW region. The matter terms are given by 
\begin{equation} 8\pi\mu_{RW}(\tau,\rho) = 3\cj^2(\tau)-\Lambda,\qquad 8\pi P_{RW}(\tau,\rho) = -2\cj'(\tau)-3\cj^2(\tau)+\Lambda. \label{rw-matter}\end{equation}
The fluid velocity vector field is given by 
\begin{equation} \vec{u}_{RW} = \pd{}{\tau}+\rho\cj\pd{}{\rho}.\label{u-rw}\end{equation}
The expansion is given by $\theta=3\cj$, and so the comment above about collapsing/expanding regions applies, with $H$ replaced by $\cj$. Likewise, the weak energy condition in the FLRW spacetime is (\ref{eq:wec}) with $H(t)$ replaced by $\cj(\tau)$. {Note that the expressions (\ref{rw-lel})-(\ref{u-rw}) can be obtained by a coordinate transformation from more familiar expressions written in co-moving coordinates - as given for example in Section 5.2 of \cite{wald1984general} - or by setting $m=0$ in the corresponding equations of Section \ref{subsec:mcvittie} above.}

We specify that both $t$ and $\tau$ increase into the future, but we do not (yet) impose any restrictions on the signs of $H$ or $\cj$. 

\section{The matching hypersurface and the preliminary matching conditions}\label{sec:matching-prelim}

We consider the following situation: there exists a hypersurface $\sigmcv$ of the McV spacetime, which separates the spacetime into two regions, $\mmcv^\choice$, so that for each choice of 1,2, $(\mmcv^\choice,g_{McV},\sigmcv)$ is a spacetime with boundary. Likewise, there exists a hypersurface $\sigrw$ of the FLRW spacetime  giving $(\mrw^\choice,g_{RW},\sigrw)$. We match $(\mmcv^\choice,g_{McV},\sigmcv)$ and $(\mrw^\choice,g_{RW},\sigrw)$ across their boundaries by constructing the manifold $\mathcal{M}=\mm^+\cup\mm^-$, where $\mm^+$ is one of $\mmcv^\choice$ and $\mm^-$ is one of $\mrw^\choice$ (four choices in all), and identifying the geometries of $\sigmcv$ and $\sigrw$.  That is, both the points and the tangent spaces of the boundary hypersurfaces are identified. (Note that we do not associate the signs $\pm$ with any notion of an interior or exterior region: the signs are simply convenient labels.) We do not specify \textit{a priori} the causal nature of $\Sigma$, and we apply the general hypersurface formalism of \cite{mars1993geometry}.

This identification of the hypersurfaces is conveniently accomplished by invoking the existence of an abstract 3-dimensional $C^3$ manifold $\sigma$, which embeds in both spacetimes with images $\sigmcv$ and $\sigrw$ respectively. See \cite{mars1993geometry}. We assume that $\sigma$ is spherically symmetric, and we introduce intrinsic coordinates 
$\xi^a=(\lambda,\theta,\phi)$ on $\sigma$ (we take $a\in\{1,2,3\}$), so that for some interval $I$,
\be \sigma=I\times\mathbb{S}^2=\{(\lambda,\theta,\phi):\lambda\in I, (\theta,\phi)\in\mathbb{S}^2\}.\label{sigma-def}\ee 

For convenience, we describe the hypersurface and associated vector and tensor fields of $(\mmcv,g_{McV},\Sigma_{McV})$, and (mostly) drop the subscripts for the time being. The corresponding formulae for $(\mrw,g_{RM},\Sigma_{RW})$ are obtained by making the substitutions 
\be (t,r,m,H(t)) \to (\tau, \rho, 0, \cj(\tau)).\label{eq:McVtoRWsubs1}
\ee
so that 
\be 
(\alpha,\beta,\gamma)\to (1-\rho^2\cj^2(\tau),\rho\cj(\tau),1).
\label{Eq:McVtoRWsubs2} 
\ee
We will use $\alpha,\beta,\gamma$ (with (\ref{al-be-ga-id}) holding) to represent generic metric functions that apply to both the McVittie and FLRW spacetimes, and as the specific values (\ref{alpha-def})-(\ref{gam-def}) relevant to McVittie only. The difference should be clear from the context. 

Then we can give the following representation of $\sig=\Sigma_{McV}$:
\begin{equation} \Sigma=\{x^\alpha\in\mm: t=\tsig(\lambda), r=\rsig(\lambda), \lambda\in I, (\theta,\phi)\in\mathbb{S}^2\}.\label{sigmcv-def}\end{equation}

The coordinate vector fields 
\begin{equation} \{\vec{e}_1= \pd{}{\lambda}, \vec{e}_2 =\pd{}{\theta}, \vec{e}_3=\pd{}{\phi}\} \end{equation} 
provide a basis for the tangent space of $\sigma$ at each point.  Using the embedding (\ref{sigmcv-def}), 
we can determine a tangent basis for $\Sigma$. In general, these vector fields have the form 
\begin{equation} \vec{e}_a = \left. e^\alpha_a\pd{}{x^\alpha}\right|_\sig,\quad e^\alpha_a=\pd{x^\alpha}{\xi^a},\end{equation}
where, as above, the embedding of $\sig$ is specified by equations of the form $x^\alpha=x^\alpha(\xi^b)$ (for ease of notation, we do not distinguish between the coordinate vector fields on $\sigma$ and their push-forward to $\Sigma$). 
Thus the basis of tangent vectors to $\Sigma$ is given by (where for convenience we introduce $\dot{t}=\dot{\tsig}$ and $\dot{r}=\dot{\rsig}$)
\begin{equation} \left.\{\vec{e}_{1} = \dot{t}\pd{}{t} + \dot{r}\pd{}{r}, \vec{e}_2,\vec{e}_3\}\right|_{\Sigma}.\label{basis-mcv}\end{equation}
The vector fields $\vec{e}_{2,3}$ play essentially no part in what follows, and so it is convenient to refer to $\vec{e}_1$ as \textit{the} tangent vector to $\sig$. 
Here and throughout, the overdot refers to differentiation with respect to the coordinate $\lambda$. We assume throughout that $\lambda\in I$, where $I$ is a maximal interval of existence of the matching hypersurface.

We introduce two other vector fields defined on the tangent space of $\mathcal{M}$ at $\Sigma$. First is the normal vector field, given by
\begin{equation} \vec{n} = \epsilon\left.\left( (-\beta\dot{t}+\gamma\dot{r})\pd{}{t} + (\alpha\dot{t}+\beta\dot{r})\pd{}{r}\right)\right|_{\Sigma},\label{norm-mcv}
\end{equation}
where we introduce $\epsilon$ with $\epsilon^2=1$. We define
\begin{equation} 
\del = \left.\alpha\dot{t}^2+2\beta\dot{t}\dot{r}-\gamma\dot{r}^2\right|_{\Sigma},\label{delmcv-def}
\end{equation}
and note that 
\begin{equation}
    g(\vec{n},\vec{n}) = -g(\vec{e}_{1},\vec{e}_{1})=\del.\label{eq:normsmcv}
\end{equation}

We will see below that $\del$ must be continuous across $\Sigma$ (\textit{cf.}\ (\ref{eq:del-cont})), and so (\ref{eq:normsmcv}) provides the following useful information, which holds independently of which embedding ($\pm$) is considered: 

\begin{lemma}\label{lemma:causal-char} The hypersurface $\Sigma$ is timelike (respectively, null, spacelike) at $p\in\sig$ if and only if $\del|_p >0$ (respectively, $\del|_p =0$, $\del|_p <0$).
\end{lemma}

We emphasize that $\vec{n}$ is determined up to a sign, corresponding to the choice of the normal as pointing out of $\mm^+$ and into $\mm^-$, or \textit{vice versa}. This choice is accommodated by the introduction of $\epsilon\in\{-1,1\}$ {in (\ref{norm-mcv}), and of a corresponding $\delta\in\{-1,1\}$ for $\sigrw$ (but note that we do not write down the corresponding expression for the normal explicitly). We are not free to choose these quantities independently (see \cite{fayos1996general}). Subject to solving the matching conditions, there are two complementary matchings allowed that correspond on the one hand to viewing the McVittie region as the exterior and the FLRW region as the interior, and on the other to viewing the McVittie region as the interior and the FLRW as the exterior. The approach below that employs a continuous rigging automatically selects the appropriate values of $\epsilon$ and $\delta$, and these play no further role.}

The hypersurface $\Sigma$ is null at points where the normal $\vec{n}$ is null - in other words, where $\del=0$. We see then that $\vec{e}_1$ is also null, and in fact $\vec{n}$ and $\vec{e}_1$ are then parallel. At such a point $p$, the tetrad $\{\vec{n}, \vec{e}_a, a=1,2,3\}$ is no longer a basis for the tangent space of $\mm$ at $p\in\Sigma$. To accommodate this situation, we introduce the second vector field mentioned above: a rigging vector $\vec{\ell}$ \cite{mars1993geometry}. This is a vector field defined on $\Sigma$ which is everywhere transverse to $\Sigma$. 

The matching conditions for a general hypersurface are expressed conveniently  using a rigging that is continuous across $\Sigma$. The necessary degree of continuity is that the spacetime norm of $\vec{l}$ is continuous across $\Sigma$, and that the tangent basis components of the rigging one-form are continuous across $\Sigma$. That is, we will seek riggings that satisfy 
\begin{equation} g^+(\vec{\ell}^+,\vec{\ell}^+) = g^-(\vec{\ell}^-,\vec{\ell}^-), \label{eq:riggin-cont-1} \end{equation}
and 
\begin{equation} \ell_a^+ =\ell_a^-,\label{eq:rigging-cont-2} \end{equation}
where 
\begin{equation} \ell_a = e^\alpha_a \ell_\alpha,\quad \ell_\alpha = g_{\alpha\beta}\ell^\beta. \label{eq:riggin-components} \end{equation}

It is convenient at this point to introduce the preliminary matching conditions, which express continuity of the spacetime metric across $\Sigma$. 

The hypersurface $\sig=\{x^\alpha\in\mm:x^\alpha=x^\alpha(\xi^b)\}$ inherits a metric from the underlying spacetime that can be represented in the form
\begin{equation} h_{ab} = g_{\alpha\beta}e^\alpha_ae^\beta_b. \label{3-metric}\end{equation}
\textit{A priori}, this means that $\sig$ may inherit different metrics when considered as the boundary hypersurface of $(\mm^+,g^+)$ and $(\mm^-,g^-)$ respectively. The \textit{preliminary matching conditions} correspond to equality of these two metrics on $\sig$:
\begin{equation} h_{ab}^+ := \left.g_{\alpha\beta}e^\alpha_ae^\beta_b\right|_{\sigmcv} = \left.g_{\alpha\beta}e^\alpha_ae^\beta_b\right|_{\sigrw} =: h_{ab}^-. \label{prelim}\end{equation} 
This yields two conditions: 
\begin{equation} 
\rsig(\lambda) = \rosig(\lambda), \label{radius-cts} 
\end{equation}
and
\begin{equation}
\delmcv(\tsig(\lambda),\rsig(\lambda)) = \delrw(\tasig(\lambda),\rosig(\lambda)),\label{del-cts}
\end{equation}
where equality holds for all $\lambda\in I$. We can express these in the more compact form 
\begin{eqnarray} 
r^+ &=& r^-,\label{eq:r-cont}\\
\del^+ &=& \del^-.\label{eq:del-cont} 
\end{eqnarray}

The preliminary matching condition (\ref{radius-cts})
allows us to adopt the single symbol $r$ to represent both variables on $\Sigma$: this will be used where convenient. Further, we can take tangential derivatives of (\ref{radius-cts}) to obtain continuity across $\Sigma$ of $\dot{r}$ and $\ddot{r}$ (we assume $C^2$ embeddings on $\Sigma$ in $\mm^\pm$):

\begin{equation} \dot{r}^+ =\dot{r}^-,\qquad \ddot{r}^+=\ddot{r}^-. \label{eq:r-derivatives-cts}\end{equation}

Likewise, tangential derivatives of $\del$ are also continuous across $\Sigma$. 

There is a useful consequence of (\ref{eq:del-cont}) which is worth flagging. From (\ref{eq:normsmcv}) and (\ref{del-cts}), we see that $\vec{e}_{1,McV}$ and $\vec{e}_{1,RW}$ have the same norm. Both are parallel to the unique tangent direction of $\sig$ orthogonal to both $\vec{e}_2$ and $\vec{e}_3$, and must yield the same sign acting on $t$ (respectively $\tau$: $t$ and $\tau$ both increase into the future). Thus we may write 
\begin{equation} \vec{e}_{1,McV}=\vec{e}_{1,RW}.\label{tgt-equal}\end{equation}
This completes the full identification of the tangent bases of $\Sigma^\pm$: we have 
\begin{equation} {\vec{e}_a}^+ = {\vec{e}_a}^-,\quad a=1,2,3. \label{eq:tgt-cont} \end{equation}

With these continuous (across $\Sigma$) functions and tangent vectors in hand, we can define the following rigging: 
\begin{equation} \vec{\ell} = \left. \left( \left( (1+\sqrt{2}\beta)\dot{t}-\sqrt{2}\gamma\dot{r}\right)\pd{}{t}+\left(-\sqrt{2}\alpha\dot{t}+(1-\sqrt{2}\beta)\dot{r}\right)\pd{}{r}\right) \right|_\Sigma. \label{eq:rig-def} \end{equation}

This is defined so that the preliminary matching conditions ensure continuity of $g(\vec{\ell},\vec{\ell})$ and $\ell_a$ across $\Sigma$, and indeed
\begin{equation} g(\vec{\ell},\vec{\ell}) = \del, \label{eq:rig-norm} \end{equation}
and 
\begin{equation} \ell_a = -\del (d\lambda)_a. \label{eq:rig-cpt}\end{equation}

But this is not the whole story. We see from (\ref{eq:rig-norm}) that $\vec{\ell}$ is null precisely at null points of $\Sigma$. Indeed we have
\be g(\vec{n},\vec{\ell}) = -\sqrt{2}\Delta, \label{eq:g-n-l} \ee 
showing that $\vec{\ell}$ fails to be a rigging at such points: we can show that $\vec{\ell}$ is parallel to $\vec{e}_1$ at null points.
This seems to defeat the purpose of introducing a rigging. However, as we will see below, we can circumvent this problem at isolated null points\footnote{Since $\Sigma$ has the form $I\times\mathbb{S}^2$, the ``isolated points" referred to are in fact 2-spheres $\{p\}\times\mathbb{S}^2$, with $p$ an isolated point of $I$. For convenience, we will continue to refer to isolated points: all of the interesting stuff happens in $I$.} of $\Sigma$ by applying a continuity argument. Using the rigging introduced here is essential to this argument. In the case where $\Sigma$ is null on an open subset, a separate approach is required - see below. {This point is worth emphasizing: we \textit{must} use the general hypersurface approach (and the associated machinery of riggings \cite{mars1993geometry}) to accommodate potential changes in causal character of the hypersurface. However, it may be the case that there is no one choice of rigging that accommodates all possibilities. For example, a given hypersurface $\Sigma$ may be null on an open subset $O_1$ with null normal given by the radial null vector field $\vec{k}$. Then the transverse radial null vector field $\vec{l}$ (with $g(\vec{k},\vec{l})\neq 0$) could be used as a rigging. But (having changed character at least twice), we may find another open subset $O_2$ of $\Sigma$ on which $\vec{l}$ is the normal, in which case $\vec{k}$ can be taken to be a rigging. Neither $\vec{k}$ nor $\vec{l}$ provides a rigging throughout $\Sigma$. See Figure 1.}

\begin{center}
    \begin{figure}
    \centering
        \includegraphics[width=10cm]{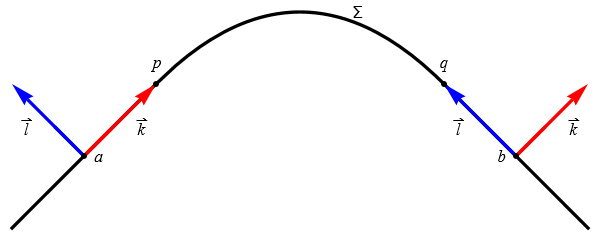}
        \caption{{A hypersurface $\Sigma$ of variable causal character in flat spacetime. The hypersurface is spacelike between $p$ and $q$, and null in neighbourhoods of the regions between $a$ and $p$ and $b$ and $q$ respectively. At $a\in\Sigma$, $\vec{k}$ is the null normal (which is also tangent to $\Sigma$ at $a$) and $\vec{l}$ is a rigging. These roles are reversed at $b\in\Sigma$. Neither $\vec{k}$ nor $\vec{l}$ is a rigging throughout $\Sigma$. 
        }}   
    \end{figure}\label{fig:hypcc}
\end{center}

\section{Junction conditions}\label{sec:junction-conditions}

The preliminary matching conditions express continuity of the spacetime metric across $\sig$. These have been considered above. The junction conditions express continuity of the (extrinsic) curvature of $\sig$ embedded in  $\mm=\mm^+\cup\mm^-$. 
For a general hypersurface, the junction conditions are expressed in terms of the \textit{rigged fundamental form}, which is defined by 
\begin{equation} Y_{ab} = e^\alpha_a e^\beta_b\nabla_\alpha \ell_\beta. 
\label{eq:rigged-ff}\end{equation}
The junction conditions for a general hypersurface are
\begin{equation} Y_{ab}^+ = Y_{ab}^-. \label{eq:junctionH}\end{equation}

Having chosen a continuous rigging (in the sense of (\ref{eq:riggin-cont-1}) and (\ref{eq:rigging-cont-2})), the preliminary matching conditions (\ref{prelim}) and the junction conditions (\ref{eq:junctionH}) provide necessary and sufficient conditions for the absence of singularities (in the form of delta functions concentrated on $\Sigma$) in the Riemann tensor of $\mm=\mm^+\cup\mm^-$. These conditions are independent of the choice of (continuous) rigging \cite{mars1993geometry}.  For clarity, we recap on the full set of conditions required for the continuous matching of two spacetimes across a shared boundary. We phrase this as a definition.  

\begin{definition}\label{def:match} Consider two spacetimes with boundaries $(\mm^+,g^+,\Sigma^+)$ and $(\mm^-,g^-,\Sigma^-)$. We say that these spacetimes \textbf{match continuously} acrosss $\Sigma^+\equiv\Sigma^-$ if and only if the following \textbf{matching conditions} hold: 
\begin{enumerate} 
\item[(i)] There is a $3-$dimensional manifold $\sigma$ and a pair of $C^3$ mappings $\Phi^\pm:\sigma\to \mm^\pm$ where $\Phi^\pm$ is a homeomorphism of $\sigma$ onto $\Sigma^\pm$  and where for each $p\in\sigma$, the push-forward \be \left.d\Phi^\pm\right|_p:T_p(\sigma)\to T_{\Phi^\pm(p)}(\mm^\pm) \label{d-phi-pm} \ee
has rank three except at isolated points of $\sigma$.
\item[(ii)] $\Sigma^+=\Phi^+(\sigma)$ and $\Sigma^-=\Phi^-(\sigma)$ are identified in a pointwise manner (so that $\Sigma^+\ni\Phi^+(p)\equiv\Phi^-(p)\in\Sigma^-$ for all $p\in\sigma$, and we can set $\Sigma:=\Sigma^+\equiv\Sigma^-$), the tangent spaces are identified as above, and these hypersurfaces are isometric in the sense that 
\be h_{ab}^+=h_{ab}^-\label{eq:isometric} \ee
everywhere on $\Sigma$.
\item[(iii)] There exist rigging vectors $\vec{\ell}^\pm$ defined on $T(\Sigma^\pm)$ which are transversal everywhere and which obey the continuity conditions 
\be g^+(\vec{\ell}^+,\vec{\ell}^+)=g^-(\vec{\ell}^-,\vec{\ell}^-),\quad \ell^+_a = \ell^-_a \label{rig-con-recap} \ee
everywhere on $\Sigma$.
\item[(iv)] The junction conditions 
\begin{equation} Y_{ab}^+ = Y_{ab}^-. \label{eq:junction-recap}\end{equation}
hold throughout $\Sigma$.  
\end{enumerate}
\end{definition}

In line with the discussion above, we refer to items (i) and (ii) here as the preliminary matching conditions, item (iii) as the rigging compatibility conditions and item (iv) as the junction conditions. As indicated, we will refer to (i)-(iv) as the matching conditions.

\begin{comments}
    As we have seen, the embeddings $\Phi^\pm$ have the general form 
    \be \Phi:\sigma = I\times\mathbf{S}^2\to \mm: (\lambda,\theta,\phi)\mapsto (t(\lambda),r(\lambda),\theta,\phi). \label{embedding} \ee
    It follows that $d\Phi$ has rank three if and only if $\dot{t}$ and $\dot{r}$ are not both zero. The rank drops to two if $\dot{t}=\dot{r}=0$, and at such a point, the tangent vanishes. This creates a fundamental problem if the maximal rank condition fails on an open subset, but does not create such difficulties if it occurs at an isolated point: the condition corresponds to a cusp of the embedded hypersurface (see Figure 2). Thus the \textit{maximal rank} condition 
    \be (\dot{t}^\pm(\lambda))^2+(\dot{r}^\pm(\lambda))^2 \neq 0,\quad \lambda\in I \label{max-rank} \ee
    must be satisfied almost everywhere on $\Sigma$. Likewise, all derivatives ($\dot{t}^\pm, \dot{r}^\pm$) must be finite. 
\end{comments}

\begin{center}
    \begin{figure}
    \centering
        \includegraphics[width=4cm]{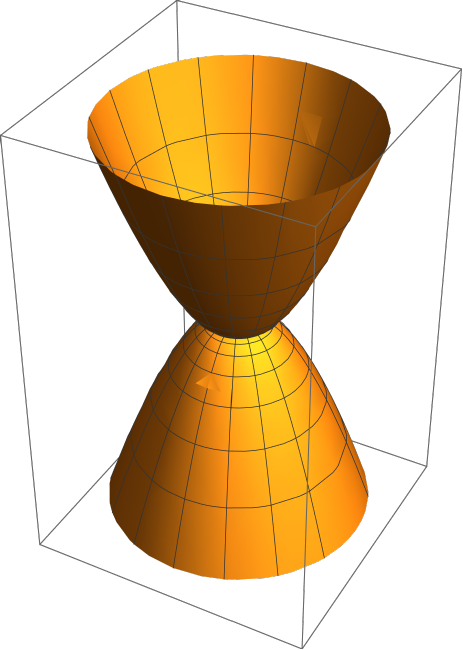}
        \caption{The figure shows the embedding $\Psi: \sigma=\real\times\mathbb{S}^1\to\real^3$ with
        $(\lambda,\phi)\in\sigma\mapsto ((1+\lambda^2)\cos\phi,(1+\lambda^2)\sin\phi,\lambda^3)$. The rank of the mapping drops from (the maximal value) two to one on the circle $S=\{x^2+y^2=1,z=0\}$. At points of this circle, the tangent $d\Psi(\partial_\lambda)$ vanishes, and every vector in the tangent space of $\real^3$ at $p\in\Psi(\sigma)$ is a normal to the surface. Geometric objects (tangents, normals, riggings) are well-defined away from $S$, and can be pushed through $S$ by continuity.}   
    \end{figure}
\end{center}

To tease out the full content of the matching conditions, there appears to be no choice but to consider a number of different cases. Before embarking on this, we restate our aims. These are to determine for which McVittie spacetimes there exists an FLRW spacetime and a matching hypersurface such that the matching conditions of Definition \ref{def:match} are satisfied. We will attempt to understand the associated global structures - of the spacetimes, and of the matching hypersurface. We will be concerned with (global) existence and uniqueness questions, and with the causal character of the matching hypersurface. 

We conclude this section by deriving the full sets of matching and junction conditions that arise in the different cases that must be considered. We defer the discussion of the solutions of these sets of equations to Section \ref{sec:solving}.

\subsection{Null matching hypersurfaces.}\label{subsec:null-matching}

We consider here the case where $\del=0$ on an open subset of $\Sigma$. Then the tangent vector $\vec{e}_1$ takes one of the forms

\begin{equation} \vec{e}_1 = \left. \dot{t}\pd{}{t}+\frac{1}{\gamma}(\beta\pm1)\dot{t}\pd{}{r}\right|_\Sigma. \label{eq:null-tgt} \end{equation}
We define a null rigging using the transverse null direction: 
\begin{equation} \vec{\ell}=\left. -(\beta\pm1)\pd{}{t}+\alpha\pd{}{r}\right|_\Sigma. \label{eq:null-rigging} \end{equation}
Then we find 
\begin{equation}{\ell_a = 2\dot{r}(d\lambda)_a.} \label{eq:null-rig-cont2}\end{equation}
Since $g(\vec{\ell},\vec{\ell})=0$, we see that both of the conditions (\ref{eq:riggin-cont-1}) and (\ref{eq:rigging-cont-2}) required for a continuous rigging are satisfied. 

{We find that in both cases $(\pm)$ of (\ref{eq:null-tgt}),} the junction condition $Y_{22}^+=Y_{22}^-$ is 
\begin{equation}\alpha^+ = \alpha^-, \label{eq:al-cts} \end{equation} so that 
\be H^2 + \frac{2m}{r^3} = \cj^2,\label{eq:al-cts1} \ee
and $Y_{11}^+=Y_{11}^-$ reads
\begin{equation}
H'(t)f^{1/2}\dot{t}^2 = \cj'(\tau)\dot{\tau}^2. \label{eq:H11-null-final-pm} \end{equation}
The radial function $r$ and its tangential derivatives are also continuous across $\Sigma$, and so using (\ref{eq:null-tgt}) we see that $\dot{r}^+=\dot{r}^-$ (i.e.\ $\vec{e}_1^+(r)=\vec{e}_1^-(r)$) reads
\be (rHf^{1/2}\pm f)\dot{t} = (r\cj\pm1)\dot{\tau}. \label{eq:null-dotr} \ee

\subsubsection{Null matching hypersurfaces of constant radius.}\label{subsubsec:null-r-con}

In the case where $r=r_0$ is constant along $\Sigma$, from $\Delta^+=\Delta^-=0$ we find 
\be \alpha\dot{t}^2 = (1-r_0^2\cj^2)\dot{\tau}^2 = 0. \label{del-null-r0} \ee
Neither $\dot{t}$ nor $\dot{\tau}$ can vanish on open subsets, and so we must have $\alpha^+=\alpha^-=0$, so that 
\be 1-\frac{2m}{r_0}-r_0^2H^2(t) = 1-r_0^2\cj^2(\tau) = 0.\label{del-null-alpha} \ee

{Then both $H=H_0$ and $\cj=\cj_0$ must be constant, and so the McVittie region is a (portion of) Schwarzschild-de Sitter spacetime, and the FLRW region is a (portion of) de Sitter spacetime. The matching is across the cosmological horizon $r=r_c$ of both spacetimes. and from (\ref{del-null-alpha}), the parameters of the spacetimes are related by} 
{\be \frac{6m}{r_c^3}+\Lambda_{\rm{SdS}} = \Lambda_{\rm{dS}}, \label{eq:SdS-dS}\ee}

{where $\Lambda_{\rm{SdS}}=3H_0^2$ and $\Lambda_{\rm{dS}}=3\cj_0^2$ are the \textit{different} cosmological constants of the two regions. But some caution is needed here in relation to riggings: the rigging (\ref{eq:null-rigging}) - which lies in the radial 2-space -  is null, and we see from (\ref{eq:null-rig-cont2}), has $\ell_a=0$ in this case. Thus this candidate for a rigging vanishes. In fact this follows from (\ref{al-be-ga-id}) (which reads $\alpha\gamma+\beta^2=1$) and $\alpha=0$. We must therefore use a different rigging. The choice}
{\be \vec{\ell} = \left.\frac{1}{\dot{t}}\left( \frac{\gamma}{2}\frac{\partial}{\partial t} +\frac{\partial}{\partial r}\right) \right|_{\Sigma} \label{eq:rig-null-rcon} \ee}
{
has the appropriate properties (we must have $\dot{t}\neq0$ along the hypersurface). This is null, and satisfies}
{
\be \ell_a = -(d\lambda)_a, \label{eq:rig-null-rcon-c2} \ee}
{
and so this choice of rigging satisfies the required continuity conditions. 
We find that}
{
\begin{align} Y_{11} &= \frac{\ddot{t}}{\dot{t}} + \frac{1}{r_c}\left(1-\frac{3m}{r_c}\right)\dot{t}, \label{eq:Y11-null-rcon} \\ 
Y_{22} &= \frac{r_c}{\dot{t}}. \label{eq:Y22-null-rcon} 
\end{align}}
{Comparing the McVittie and FLRW values of $Y_{22}$ gives $\dot{t}=\dot{\tau}$. Then the second derivative terms in $Y_{11}^+=Y_{11}^-$ cancel, and we must have $m=0$. Thus the result is two regions of de Sitter spacetime with the same cosmological constant, a trivial situation.}

\subsection{Non-null matching hypersurfaces.}\label{subsec:non-null-matching}

We deal here with general matching hypersurfaces that may change their causal character, but do so only at isolated null points. There is a further distinction between hypersurfaces on which the tangential derivative $\dot{r}$ vanishes only at isolated points and hypersurfaces admitting open subsets on which $r|_\Sigma$ is constant - i.e.\ constant radius hypersurfaces. We deal with this case first.

\subsubsection{Non-null hypersurfaces of constant radius.}\label{subsubsec:non-null-r-constant}

In this case we have $\del\neq0$ except possibly at isolated points, and 
\begin{equation} r=r_0 \label{eq:r0} \end{equation}
where $r_0$ is constant. $\del$ is continuous across $\Sigma$, along with its tangential derivatives. 
The tangent to $\Sigma$ in the radial 2-space is $\vec{e}_1=\dot{t}\partial_t$, and for the rigging, we can use the normal vector 
\begin{equation} \vec{n} = \left.-\beta\dot{t}\pd{}{t}+\alpha\dot{t}\pd{}{r}\right|_\Sigma .\label{eq:r0-normal} \end{equation} 
We find $\del=\alpha\dot{t}^2$, and so we have 
\begin{equation} \alpha^+(\dot{t}^+)^2 = \alpha^-(\dot{t}^-)^2.\label{eq:r0-perlim1} \end{equation}
The normal is also a rigging for $\Sigma$, and we have
\be g^\pm(\vec{l}^\pm,\vec{l}^\pm)=\alpha^\pm(\dot{t}^\pm)^2 = \Delta^\pm, \label{gll-r-constant} \ee
and 
\be \ell_a^\pm =0, \label{ella-pm} \ee
so that the rigging compatibility conditions are satisfied. The junction conditions are (\ref{eq:junctionH}):
\begin{eqnarray}
Y_{11}^+ &=& Y_{11}^-,\label{H11}\\
Y_{22}^+ &=& Y_{22}^-. \label{H22}
\end{eqnarray}

The latter yields 
\begin{equation} \alpha^+\dot{t}^+ = \alpha^-\dot{t}^-. \label{eq:r0-jct1} \end{equation}
Since we are at a non-null point of $\Sigma$ so that $\del=\alpha\dot{t}^2\neq0$, neither $\alpha$ nor $\dot{t}$ can vanish, and so from (\ref{eq:r0-perlim1}) and (\ref{eq:r0-jct1}) we have 
\begin{equation} \alpha^+ = \alpha^-,\quad \dot{t}^+ = \dot{t}^-. \label{eq:r0-main} \end{equation} 
The first of these is equivalent to 
\begin{equation} H^2(t)+\frac{2m}{r_0^3}  = \cj^2(\tau), \label{eq:r0-mass-match}\end{equation}
and by the second, time derivatives also match, giving 
\begin{equation} H(t)H'(t) = \cj(\tau)\cj'(\tau). \label{eq:r0-HJ-match} \end{equation}
At isolated null points, these also hold by continuity. 

We find 
\begin{equation} Y_{11}=-\frac12\dot{t}^3(\beta\partial_t \alpha+\alpha\partial_r\alpha-2\alpha\partial_t\beta), \label{eq:r0-H11} \end{equation}
and so $Y_{11}^+=Y_{11}^-$ yields 
\begin{eqnarray} (\frac{2m}{r_0^2}-2r_0H^2(t))(1-\frac{2m}{r_0}-r_0^2H^2(t)) - 2r_0(1-\frac{2m}{r_0})^{1/2}H'(t)  && \nonumber\\  = -2r_0\cj^2(\tau)(1-r_0^2\cj^2(\tau)) -2r_0\cj'(\tau). &&
\label{eq:r0-main1} \end{eqnarray}

\subsubsection{Non-null hypersurfaces of non-constant radius.}\label{subsubsec:non-null-r-noncon}

In this case, we use the rigging (\ref{eq:rig-def}) introduced in Section \ref{sec:matching-prelim} above. Define 
\begin{equation} \Gamma=\alpha\dot{t}+\beta\dot{r}.\label{eq:Gam-def} \end{equation} Using (\ref{eq:r-derivatives-cts}), we find that $[Y_{22}]=0$ is equivalent to 
\begin{equation} \Gamma^+ = \Gamma^-, 
\label{jct22}
\end{equation}
which can be written 
\begin{equation} \left.\alpha\dot{t}+\beta\dot{r}\right|_\Sigma = \left.(1-r^2\cj^2)\dot{\tau}+r\cj\dot{r}\right|_\Sigma. \label{eq:H22-full}\end{equation}

To consider the junction condition $Y_{11}^+=Y_{11}^-$,  we deal first with points of $\Sigma$ where $\dot{r}\neq 0$. At such points, we find that 
\begin{equation} Y_{11} = \sqrt{2}\frac{\del}{\dot{r}}\dot{\Gamma}-\frac12\left(1+\sqrt{2}\frac{\Gamma}{\dot{r}}\right)\dot{\del} - \frac{1}{\sqrt{2}}\frac{\del}{\dot{r}}(\partial_t\alpha\dot{t}^2+2\partial_t\beta\dot{t}\dot{r}). \label{eq:H11a} \end{equation}
Invoking the continuity of $r,\del$ and $\Gamma$ and their tangential derivatives, using (\ref{alpha-def})-(\ref{gam-def}) and $\Delta\neq0$, we find that $[Y_{11}]=0$ is equivalent to 
\begin{equation} \Pi^+ = \Pi ^-, \label{eq:pi-cont} \end{equation}
where 
\begin{equation} \Pi = \gamma^{-1/2}H'\dot{t}(\beta\dot{t}-\gamma\dot{r}). \label{eq:pi-def} \end{equation}

If $\dot{r}=0$ at an isolated point of $\Sigma$, then $[Y_{11}]=0$ again leads to (\ref{eq:pi-cont}) by continuity. 

Thus at non-null points of $\Sigma$, the full matching conditions consist of (\ref{eq:r-cont}), (\ref{eq:del-cont}), (\ref{jct22}) and (\ref{eq:pi-cont}). For $C^1$ metric functions and $C^2$ embeddings (the latter amounting to the requirement that the mapping $\lambda\mapsto (t_\Sigma(\lambda),r_\Sigma(\lambda),\tau_\Sigma(\lambda),\rho_\Sigma(\lambda))$ is $C^2$), these conditions extend to isolated null points of $\Sigma$ by continuity.

One consequence of these conditions is worth flagging at this point. We have the identity 
\begin{equation}\Gamma^2 = \alpha\del + \dot{r}^2. \label{eq:Gam-id} \end{equation}
Then the preliminary matching conditions $\del^+=\del^-$ and $\dot{r}^+=\dot{r}^-$, the junction condition $\Gamma^+=\Gamma^-$ and the non-null condition $\del\neq0$ lead to 
\begin{equation}\alpha^+ = \alpha^-.\label{eq:alpha-cts-again} \end{equation}
This holds at non-null points and extends to isolated null points by continuity.

\section{Consequences of matching}\label{sec:consequences}

In Section \ref{sec:solving}, we will address the questions of existence and uniqueness of matching hypersurfaces. Before doing so, we consider some physical and geometric consequences of matching. 

\subsection{Continuity of the mass}\label{subsec:conseq-mass}

In spherically symmetric spacetimes, the Misner-Sharp mass $E$ is defined by 
\begin{equation} E= \frac{r}{2}(1-\chi),\quad \chi=g^{\alpha\beta}\nabla_\alpha r\nabla_\beta r.\label{ems-def}\end{equation} 
This has a number of properties that lead to its characterization as the gravitational energy in spherical symmetry - see \cite{hayward1996gravitational}. This quantity is sometimes referred to as the Hawking mass, being the spherically symmetric case of the more general mass function introduced in \cite{hawking1968gravitational}. The quantity $E$ seems to have made its first appearence in \cite{misner1964relativistic}. 

For the McVittie family of spacetimes with line element (\ref{mcv-lel}), we find $\chi=\alpha$ and so 
\begin{equation}E_{McV} = m + \frac{r^3}{2}H^2 = m + \frac43\pi r^3\mu_{McV} + \frac{\Lambda}{6}r^3. \label{eq:E-mcv} \end{equation}
Thus we see that the matching condition $\alpha^+=\alpha^-$, which arises in all cases of $\Sigma$ (non-null, null, constant or non-constant radius) is equivalent to the continuity across $\Sigma$ of the Misner-Sharp mass, $E^+=E^-$. This equation can be written 
\begin{equation} \left.m+\frac{r^3}{2}H^2\right|_\Sigma = \left.\frac{r^3}{2}\cj^2\right|_\Sigma.\label{ems-cts}\end{equation}

This equation may be considered to be the equation of motion of $\sig$: we can rearrange (and change the emphasis) to obtain 
\begin{equation}\rsig^3(\lambda) = \frac{2m}{\cj^2(\tasig(\lambda))-H^2(\tsig(\lambda))}. \label{eom} \end{equation}
Furthermore, this equation has a direct physical interpretation. Using (\ref{mcv-matter}) and (\ref{rw-matter}), (\ref{ems-cts}) may be written as
\begin{equation} \left.\frac{m}{r^3} + \frac{4\pi}{3}\mu_{McV}(t)\right|_\Sigma = \left. \frac{4\pi}{3}\mu_{RW}(\tau)\right|_\Sigma.\label{overdensity} \end{equation}

This indicates that, with $m>0$, the FLRW region is over-dense with respect to the McVittie region. Pending the successful construction of the matching hypersurface, this adds further to the interpretation of the McVittie metric: this represents a solution of the Einstein equations with locally spatially homogeneous energy density, which (i) is exterior to a central over-dense, isotropic interior region or (ii) is a central under-dense local system, interior to an isotropic exterior. These may be considered to be analogous to the Oppenheimer-Snyder and the Einstein-Straus constructions respectively, but with non-zero density throughout the spacetime. See also \cite{haines1993thin} for a corresponding interpretation in terms of shell quantities, potential energies and mass deficits in the FLRW/shell/McVittie model.

\subsection{Radial null geodesics, trapping and horizons.}\label{subsec:conseq-trapping}
The mass continuity equation (\ref{ems-cts}) may be written as
\begin{equation} 1-\frac{2m}{r}-r^2H^2 = 1-r^2\cj^2.\label{m1c-alt}\end{equation}
These quantities ($\alpha=1-2m/r-r^2H^2$ in the McVittie region, $1-r^2\cj^2$ in the Robertson-Walker region) mark the location of horizons in the respective spacetimes: 

The trapped nature of a region of the spacetime is determined by the signs of the null expansions $\theta^{(1,2)}$ defined  by 
\begin{equation} {\cal{L}}_{{\vec{n}}^{(1,2)}}\omega  = \theta^{(1,2)}\omega, \label{null-exp} \end{equation}
where ${\cal{L}}$ is the Lie derivative, $\omega = r^2 \sin\theta d\theta \wedge d\phi$ is the area 2-form of the 2-spheres of the spacetime (i.e.\ the orbits of the $SO(3)$ symmetry group) and ${\vec{n}}^{(1,2)}$ are the unique (up to scaling by positive factors) future-pointing null directions orthogonal to the 2-spheres (we will refer to these as radial null directions).  The signs of $\theta^{(1,2)}$ are therefore invariant. In fact these null directions are tangent to null geodesics, and choosing a scaling that corresponds to affine parametrisation, we have the equivalent definition in terms of null divergences: $\theta^{(1,2)}=\nabla\cdot{\vec{n}}^{(1,2)}$. We can take (with the upper sign corresponding to the index $(1)$, and the lower sign to $(2)$)
\begin{equation} \vec{n}^{(1,2)} = \gamma\pd{}{t}+(\beta\pm1)\pd{}{r}, \label{null-dirs-mcv}\end{equation}
which gives
\begin{equation} \theta^{(1,2)} = \frac{2}{r}(\beta\pm1). \label{null-exp-mcv} \end{equation}
Then 
\begin{equation} \theta^{(1)}\theta^{(2)} = \frac{4}{r^2}(\beta^2-1) = -\frac{4\alpha\gamma}{r^2}. \label{alpha-null-exp} \end{equation}
These formulas apply in both the McVittie region and the FLRW region with the appropriate substitutions for the metric functions (and time coordinate). 

By definition, a region of spacetime is untrapped (or regular) if and only if $\theta^{(1,2)}$ have opposite signs, giving $\alpha>0$. (Recall that $\gamma=(1-2m/r)^{-1}>0$ in the McVittie region, and $\gamma=1$ in the FLRW region.) A region is trapped if the null expansions are both negative (yielding $\alpha<0$): note that this can only occur if the region is collapsing ($H<0$ in McVittie, {$\cj<0$} in FLRW), and a region is anti-trapped if the null expansions are both positive (also yielding $\alpha<0$), which can occur only during an expanding phase $(H>0, {\cj>0})$. 

We define a horizon $\hor$ to be a hypersurface foliated by marginally trapped 2-spheres, so that either $\theta^{(1)}$ or $\theta^{(2)}$ vanishes throughout $\hor$. Thus $\alpha=0$ on a horizon of McVittie spacetime, and $1-r^2\cj^2=0$ on a horizon of a Robertson-Walker spacetime. 

Then (\ref{ems-cts}) has the additional interpretation that the regular region of McVittie matches across $\Sigma$ with the regular region of the FLRW spacetime \cite{fayos1996general}. A trapped region cannot match to an untrapped region, and furthermore, can only match with a trapped region \cite{fayos1996general} - that is, not with an anti-trapped region.

\subsection{Oppenheimer-Snyder}\label{subsec:conseq-OS}

Schwarzschild-de Sitter spacetime is a special case of a McVittie spacetime, and so it should be possible to recover the Oppenheimer-Snyder model in the framework being considered here. Recall that this model entails a collapsing homogeneous dust sphere with Schwarzschild exterior \cite{oppenheimer1939continued}. The sphere collapses through the Schwarzschild radius and down to zero radius.  The approach taken in \cite{oppenheimer1939continued} is to consider a spherically symmetric spacetime filled with inhomogeneous dust. This model is fully described by two free functions of a co-moving radial coordinate $R$, the mass and the energy functions (see e.g.\ Section 22.7 of \cite{griffiths2009exact} for a brief summary). In \cite{oppenheimer1939continued}, these are chosen to correspond to an initially \textit{homogeneous} dust sphere for $R\leq R_0$, and vacuum for $R>R_0$. These conditions are shown to maintain in the evolution, and a condition that the sphere is initially collapsing also continues to hold. The collapse proceeds to zero radius as measured by an observer co-moving with the dust particles. 

We recover this model by considering the most general matching of a Schwarzschild-de Sitter spacetime with a (spatially flat) FLRW spacetime across a non-null hypersurface. {(Recall that the case of matching across a null hypersurface was dealt with in Section \ref{subsec:null-matching}.)} Thus we take $H=H_0$ to be a non-negative constant (with the cosmological constant $\Lambda=3H_0^2$), and so 
\be \alpha^+ =1-\frac{2m}{r}-r^2H_0^2,\quad \beta^+ = rH_0(\gamma^+)^{1/2},\quad \gamma^+=(1-\frac{2m}{r})^{-1},\label{SdS-metric} \ee
and 
\be \alpha^-=1-r^2\cj^2(\tau),\quad\beta^-=r\cj(\tau),\quad \gamma^-=1. \label{FLRW-for-SdS} 
\ee 

Since $\Sigma$ is assumed non-null, the matching and junction conditions of Section \ref{subsec:non-null-matching} hold. {We can rule out the case of a constant radius matching hypersurface as follows. Assume that $r=r_0=$ constant on $\Sigma$. Then with} $H=H_0$ constant, (\ref{eq:r0-mass-match}) gives $\cj=$ constant. Since $\alpha^+=\alpha^-\neq0$ (see (\ref{eq:r0-main1}) and the preceding comments), equation (\ref{eq:r0-main}) yields an equation that, along with (\ref{eq:r0-mass-match}), gives $m=0$, contradicting $m>0$. Thus the relevant matching conditions are those of Section \ref{subsubsec:non-null-r-noncon}, {and (in particular) $r$ is not constant on (open subsets of) $\Sigma$}. 

These are equivalent to (\ref{jct22}), (\ref{eq:pi-cont}) and (\ref{eq:alpha-cts-again}) which can be written as, respectively, 
\begin{align} 
(1-\frac{2m}{r}-r^2H_0^2)\dot{t}+rH_0(1-\frac{2m}{r})^{-1/2}\dot{r} &= (1-r^2\cj^2)\dot{\tau}+r\cj\dot{r},\label{OS-m1} \\
0 &= \cj'\dot{\tau}(r\cj\dot{\tau}-\dot{r}),\label{OS-m2} \\
H_0^2+\frac{2m}{r^3} &= \cj^2. \label{OS-m3} 
\end{align} 

{Two possibilities arise from (\ref{OS-m2}). If $\cj'\dot{\tau}=0$, then the tangential derivative of (\ref{OS-m3}) gives $\dot{r}=0$. Hence $r=r_0 =$ constant along the matching hypersurface, contradicting an earlier assumption. Thus the second possibility of (\ref{OS-m2}) holds, giving}

\be r\cj\dot{\tau}-\dot{r}=0. \label{OS-comove} \ee
As we will see in the next section, this is the condition that $\Sigma$ is co-moving in the FLRW spacetime. Then taking the tangential derivative of (\ref{OS-m3}) and substituting for $\dot{r}$ yields 
\be -2\cj'-3\cj^2+\Lambda=0,\label{OS-pressure} \ee
which is the condition that the FLRW region is pressure free (\textit{cf.} (\ref{rw-matter})). The co-moving condition is readily solved to give 
\be r(\tau) = r_0\left(\frac{a(\tau)}{a_0}\right),\label{OS-r} \ee
where $a$ is the scale factor of the FLRW spacetime. The remaining equation (\ref{OS-m1}) relates the rate of change of the time coordinate $t$ of the Schwarzschild-de Sitter region to the proper time of the FLRW region. We have $\Delta=\dot{\tau}^2$, so the matching hypersurface is timelike everywhere. Taking the FLRW region to correspond to $\{r<r(\tau)\}$ gives the Oppenheimer-Snyder model: with the FLRW region occupying $\{r>r(\tau)\}$, we obtain the Einstein-Straus model \cite{einstein1945influence}. 

{Restoring $\Lambda=3H_0^2$, we can solve (\ref{OS-pressure}) to obtain, for $H_0\neq0$, 
\be \cj(\tau) = H_0\coth\left(\frac32H_0\tau\right),\quad \tau \in (0,+\infty). \label{OS-J-sol1} \ee
For $H_0=0$, so that $\Lambda=0$ and the McVittie region is a portion of Schwarzschild spacetime, we have 
\be \cj(\tau) = \frac23\tau^{-1},\quad \tau \in ( 0,+\infty). \label{OS-J-sol2} \ee
A translation of $\tau$ is used to set the Big Bang singularity at $\tau=0$. Then (\ref{OS-m3}) can be rearranged to give an alternative (but equivalent) expression for $\left.r\right|_\Sigma$. Notice that the mass parameter $m$ of the Schwarzschild-de Sitter region appears only in the expression for the radius of the matching hypersurface, but not in the scale factor of the FLRW region of the matched spacetime.}

We can summarize as follows. 

\begin{proposition}\label{prop:OS}
    Let $(M^+,g^+$) be a Schwarzschild-de Sitter spacetime and let $(M^-,g^-)$ be a spatially flat FLRW spacetime, {and let these spacetimes be matched, in the sense of Definition IV.1, across a non-null hypersurface $\Sigma$}. Then the FLRW spacetime has zero pressure, and the matching hypersurface is co-moving in the FLRW spacetime. This corresponds to the Oppenheimer-Snyder/Einstein-Straus models. 
\end{proposition}

\subsection{Matter terms and co-moving matching hypersurfaces.}\label{subsec:conseq-matter}

Next, we consider the consequences of matching for the continuity of the matter terms. Here, we restrict to the case where $\Sigma$ is timelike. Then we can take the normal to $\Sigma$ (\ref{norm-mcv}) as a rigging, and the rigged fundamental form is the second fundamental form (or extrinsic curvature) $K_{ab}$ of $\Sigma$. The consequences in question arise through the Gauss-Codazzi equations, which show that the quantities $G_{\alpha\beta}e^\alpha_an^\beta$  and $G_{\alpha\beta}n^\alpha n^\beta$ may be written in terms of $K_{ab}$ and its covariant derivatives (with respect to the Levi-Civita connection of $h_{ab}$) \cite{mars1993geometry}. The matching conditions (and underlying assumptions on the geometry of ${\cal{M}}$ and $\sig$) show that these are independent of the embedding and so 
\begin{eqnarray}
\left.G_{\alpha\beta}e^\alpha_an^\beta\right|_{\sigmcv} &=& \left. G_{\alpha\beta}e^\alpha_an^\beta\right|_{\sigrw}, \label{gtn-cts}\\
\left.G_{\alpha\beta}n^\alpha n^\beta\right|_{\sigmcv} &=& \left. G_{\alpha\beta}n^\alpha n^\beta\right|_{\sigrw}. \label{gnn-cts}
\end{eqnarray}

Applying the Einstein equations for a perfect fluid (which apply in both spacetime regions), we obtain
\begin{eqnarray} \left.(\mu+P)(u\cdot e_1)(u\cdot n)\right|_{\sigmcv} &=& \left.(\mu+P)(u\cdot e_1)(u\cdot n)\right|_{\sigrw}, \label{matter1}\\
\left.(\mu+P)(u\cdot n)^2 + P\right|_{\sigmcv} &=& \left.(\mu+P)(u\cdot n)^2 + P\right|_{\sigrw}, \label{matter2}
\end{eqnarray}

We emphasize that these are not new conditions that must be satisfied by the various terms on $\Sigma$: rather, they follow from the conditions previously obtained. Some no-go results emerge from these equations when we consider configurations where the matching hypersurface is \textit{co-moving} with the fluid of the McVittie region or of the FLRW region or of both. 

Two of the three cases in Proposition \ref{prop:co-moving} below immediately yield no-go results that rule out non-trivial matching. Here, non-trivial means that the McVittie spacetime is neither an FLRW spacetime ($m=0$), nor a Schwarzschild-de Sitter spacetime ($H=H_0=$ constant). The results rule out certain `natural' matching scenarios, where there is no leakage of the fluid from one region into the other. In the third case, a one-parameter family of configurations is found. 

The matching hypersurface $\Sigma$ is co-moving with the fluid if the fluid flow vector is tangent to $\Sigma$, so that 
\be {\left.(u\cdot n)\right|_{\Sigma} = \epsilon\gamma^{-1/2}(\beta\dot{t}-\gamma\dot{r}) = 0.} \label{eq:co-moving} \end{equation}
Notice that if this term vanishes on both sides of (\ref{matter2}), the pressure continuity condition familiar from the Oppenheimer-Snyder model arises. From (\ref{eq:co-moving}), we see that $\dot{r}$ can vanish only at isolated points, and since $\Delta>0$, the assumptions of Section \ref{subsubsec:non-null-r-noncon} above apply, along with the relevant matching and junction conditions. {We can give the following complete description of matching configurations in which one or both of the embedded hypersurfaces are co-moving with the fluid.} {\sout{The statements below arise from these. We won't consider these cases any further: the key point here is that assuming a co-moving condition on either or both sides of the matching hypersurface is extremely restrictive.}} 

\begin{proposition}\textbf{Co-moving matching hypersurfaces.}\label{prop:co-moving}
\begin{itemize}
    \item[(i)] If $\Sigma$ is co-moving in both the McVittie region and the FLRW region, then $m=0$ (so that the McVittie region is isotropic) and the two regions have the same Hubble function: $H(t)=\cj(\tau)$ and, without loss of generality, $t=\tau$.
    \item[(ii)] If $\Sigma$ is co-moving in the FLRW region but not the McVittie region, then $H=H_0=$ constant, and so the McVittie region is (a portion of) Schwarzschild-de Sitter spacetime with cosmological constant $\Lambda=3H_0^2$. {This situation corresponds to the Oppenheimer-Snyder model considered above and the conclusions of Section \ref{subsec:conseq-OS} apply.}
        \item[(iii)] If $\Sigma$ is co-moving in the McVittie region but not in the FLRW region, then the FLRW region is a portion of de Sitter spacetime with $\Lambda=3\cj_0^2$ where, {without loss of generality}, $\cj_0>0$. The Hubble function of the McVittie region satisfies 
    \begin{eqnarray}
H' &=& \frac32(\cj_0^2-H^2)\left[1-(2m)^{2/3}(\cj_0^2-H^2)^{1/3}\right]^{1/2}. \label{Hub-eom}\end{eqnarray}

\end{itemize}
\end{proposition}

{
\noindent\textbf{Proof:} The matching and junction conditions in this case are those of Section \ref{subsubsec:non-null-r-noncon} which can be written as 
\begin{align} 
[\Gamma]=0:\quad & \alpha\dot{t}+\beta\dot{r} = (1-r^2\cj^2)\dot{\tau}+r\cj\dot{r},\label{com-m1} \\
[\Pi]=0:\quad & \gamma^{1/2}H'\dot{t}(\beta\dot{t}-\gamma\dot{r}) = \cj'\dot{\tau}(r\cj\dot{\tau}-\dot{r}),\label{com-m2} \\
[\alpha]=0:\quad &H^2+\frac{2m}{r^3} = \cj^2. \label{com-m3} 
\end{align} 
These are, respectively, (\ref{jct22}), (\ref{eq:pi-cont}) and (\ref{eq:alpha-cts-again}). It will be convenient to label quantities from the McVittie region as $\alpha^+$ etc, and quantities from the FLRW region as $\alpha^-$ etc.} 

{
In case (i), we can use (\ref{eq:co-moving}), the definition (\ref{eq:Gam-def}) and the junction condition (\ref{com-m1}) to obtain 
\be \Gamma^+ = \gamma^{-1}\dot{r} = \Gamma^- =\dot{r}, \label{com-i-1} \ee
which immediately yields $m=0$ and, in addition $H\dot{t}=\cj\dot{\tau}$. (\ref{com-m2}) is satisfied identically since the co-moving condition holds in both $(\pm)$ regions, and (\ref{com-m3}) with $m=0$ yields $H=J$. The conclusion of (i) follows. 
}

{In case (ii), the co-moving assumptions imply that $\Pi^-=0$, and since $\beta\dot{t}-\gamma\dot{r}\neq0$, (\ref{com-m2}) yields $H'=0$. It follows that $H(t)=H_0=$ constant, and the conclusion follows.}

{In case (iii), the co-moving assumptions imply that $\Pi^+=0$, and since $r\cj\dot{\tau}-\dot{r}\neq0$, (\ref{com-m2}) yields $\cj'=0$. It follows that $\cj(\tau)=\cj_0=$ constant, and the FLRW region is a portion of de Sitter spacetime with cosmological constant $\Lambda=3\cj_0^2$. Then taking the tangential derivative of (\ref{com-m3}), using the co-moving condition $\beta\dot{t}=\gamma\dot{r}$ and $\dot{t}\neq0$ yields 
\be H'(t) = \frac{3m}{r^3}\gamma^{-1/2}.\label{comiii-1} \ee
Using (\ref{com-m3}) again to remove the occurrences of $r$ in this last equation yields the ODE (\ref{Hub-eom}). Note that this equation is invariant under $\cj_0\to-\cj_0$, and so solutions with $\cj_0<0$ are equivalent to those with $\cj_0>0$. \hfill$\blacksquare$}

{We see in Case (iii) that a 3-parameter family of configurations arises, the parameters being the mass parameter $m$ and the cosmological constant of the de Sitter region. We will not consider the (global) nature of the resulting McVittie spacetimes: we conclude this section by noting that the co-moving assumption is highly restrictive, reducing the McVittie state space (two parameters $m,\Lambda$ and the function $H(t)$) to this 2-parameter class.}

\section{Solving the matching conditions}\label{sec:solving}

Having identified the complete set of matching conditions in Section \ref{sec:junction-conditions} above, we now turn to solving these. 

In solving the matching conditions, we are addressing the question of what Hubble functions will allow matching: which McVittie spacetimes can be matched to a (suitable) FLRW universe? 

To solve the matching conditions, we require a McVittie spacetime characterised by mass parameter $m$ and Hubble function $H:I_H\to\mathbb{R}:t\mapsto H(t)$, a Robertson-Walker universe characterised by the Hubble function $\cj:I_\cj\to\mathbb{R}:\tau\mapsto \cj(\tau)$, and a matching hypersurface $\Sigma$, characterised by functions $\rsig,\tsig$ and $\tasig$, such that the matching conditions described in Definition \ref{def:match} above are satisfied. We will refer to a collection
\be \{I, (r_\Sigma,t_\Sigma,\tau_\Sigma)\in C^1(I,\real^3), H\in C^1(I_H,\real), \cj\in C^1(I_\cj,\real)\} \label{config} \ee
that satisfies the matching conditions as a \textbf{matching configuration} for the problem. $I$ is an interval, and we can take $I_H$ and $I_\cj$ to be the images of $I$ under $\lambda\mapsto t(\lambda)$ and $\lambda\mapsto\tau(\lambda)$ respectively. The abstract matching hypersurface is $\sigma=I\times\mathbb{S}^2$.

We consider the three cases identified above separately, dealing first with the special cases of non-null constant radius hypersurfaces and null hypersurfaces. The fourth case of a null hypersurface of constant radius has been ruled out. 

\subsection{Null hypersurfaces of non-constant radius.}\label{subsec:sol-null-r-noncon}

In this case, the full set of matching conditions comprises 
\begin{equation}\del^+=\del^-=0,\quad r^+=r^-, \quad \alpha^+=\alpha^- \label{eq:match-null-non-const}\end{equation}
along with (\ref{eq:H11-null-final-pm}):
\be H'(t)f^{1/2}\dot{t}^2 = \cj'(\tau)\dot{\tau}^2.\label{sol-null-H11} \ee 
The null condition yields (\textit{cf.} (\ref{eq:null-tgt})) 
\be \dot{r} = (rHf^{1/2}\pm f)\dot{t}=(r\cj\pm1)\dot{\tau}. \label{eq:null-tgt-match} \ee

We consider the upper sign in this equation and assume an expanding McVittie region so that $H>0$. The alternatives yield similar conclusions. 

Since $\dot{r}$ is non-vanishing, we have $\dot{t}\neq0$ on $\Sigma$ and then the first equation of (\ref{eq:null-tgt-match}) provides the equation of motion of $\Sigma$ in the McVittie region in the form $dr/dt=rHf^{1/2}+f$.  

The condition $[\alpha]=0$ can be written $2m+r^3H^2=r^3\cj^2$. Differentiating with respect to $t$ and using (\ref{eq:null-tgt-match}) to eliminate $\dot{t}$ and $\dot{\tau}$ and then dividing by the non-zero $\dot{r}$ yields 
\be \frac{2rHH'}{rHf^{1/2}+f}+3H^2 = \frac{2r\cj\cj'}{r\cj+1}+3\cj^2. \label{eq:6-1-main1} \ee
We use (\ref{eq:null-tgt-match}) again to eliminate $\dot{t}$ and $\dot{\tau}$ from (\ref{sol-null-H11}), and use the result to eliminate $\cj'$ from (\ref{eq:6-1-main1}), and use $[\alpha]=0$ to eliminate $\cj$. This results in an ODE of the form $Q(r(t),H(t),H'(t))=0$.
{Explicitly, we have the following 2-dimensional system that describes the dynamics of a McVittie spacetime (i.e.\ the dynamics of $H(t)$) and of a null hypersurface embedded in the spacetime. It is slightly more convenient to write this is terms of 
\be z=\left(1-\frac{2m}{r}\right)^{1/2} \in (0,1). \label{z-def} \ee
Then the equation of motion of the null hypersurface and the evolution equation $Q=0$ mentioned above are, respectively,
\begin{align}
    \frac{dz}{dt} &= \frac{H}{2}(1-z^2)+\frac{z}{2m}(1-z^2)^2,\label{eom-z} \\
    \left((1-z^2+\frac{4m^2}{(1-z^2)^2}H^2)^{1/2}+1-z^2-\frac{2mz}{1-z^2}H\right)\frac{dH}{dt} &= - \frac{3}{8m^2}(1-z^2)^3\left(\frac{2m}{1-z^2}+z\right)^2.\label{eom-Hz} 
\end{align}
Given a solution of this system on an interval $I$, we can use $[\alpha]=0$ to determine $J(t)=\cj(\tau(t)$, and then use (\ref{eq:null-tgt-match}) to determine $\tau(t)$. Note that in this approach, it makes sense to consider $t$ to be the parameter along $\Sigma$. This is possible since we must have $\dot{t}\neq0$ on $\Sigma$. Short-time existence of the matching configuration is essentially immediate given initial values $r_0=r(t_), H_0=H(t_), \cj_0=J(t_0)$ and $\tau_0=\tau(t_0)$ at some $t_0\in\real$, where the values must be chosen to ensure that the coefficient on the left hand side of (\ref{eom-Hz}) is non-zero. As in Case (iii) of Proposition \ref{prop:co-moving}, we emphasize that this yields a highly restricted class of McVittie spacetimes. We will not consider the physical significance of these configurations any further, but it is perhaps worth noting that (\ref{eom-Hz}) guarantees that the weak energy condition $H'(t)<0$ is satisfied (see Comment 3(iii) below). The key point here is that it is not possible to match a \textit{general} McVittie spacetime across a null hypersurface to an FLRW spacetime.}

\subsection{Non-null hypersurfaces of constant radius.}\label{subsec:sol-non-null-r-con}
The full matching conditions in this case are  
\begin{equation}r^+=r^-=r_0>2m, \quad  \alpha^+=\alpha^-, \quad \dot{t}^+=\dot{t}^- \label{eq:non-null-r0-match-full}\end{equation} along with  (\ref{eq:r0-main1}). 
{As in the case of matching across a null hypersurface of constant radius, we can fully solve these matching conditions. Using the three equations just mentioned we can derive the ODEs 
\begin{align}
    H' &= \frac32J(H+f_0^{1/2}J),\label{eom-H-nonnull-rcon}\\
    J' &= \frac32H(H+f_0^{1/2}J),\label{eom-J-nonnull-rcon} 
\end{align}
where by (\ref{eq:r0-main}) we can assume that $\tau=t$ and we then use $J(t)=\cj(\tau)$. We also use $f_0=1-2m/r_0$. Defining
\be K= H+f_0^{1/2}J,\quad L = f_0^{1/2}H+J \label{KL-def} \ee
we find 
\be K' =\frac32 KL,\quad L'=\frac32 K^2, \label{dK-dL} \ee
and from (\ref{eq:r0-mass-match}), 
\be L^2-K^2 = \frac{4m^2}{r_0^4}.\label{K-L-mass-con} \ee
Then we can solve for $L$ and $K$ explicitly in $t$, and we find two families of solutions ($\pm$) for $H,J$: 
\begin{align}  
H_\pm(t) &= \frac{1}{r_0\sinh(t/t_0)}(f_0^{1/2}\cosh(t/t_0)\pm1), 
\label{H-nonnull-rcon-sol} \\
J_\pm(t) &= -\frac{1}{r_0\sinh(t/t_0)}(\cosh(t/t_0)\pm f_0^{1/2}). \label{J-nonnull-rcon-sol}
\end{align}
We use $t_0=r_0^2/(3m)$. In both cases, both scale factors are singular at $t=0$ (there is an unavoidable singularity, which is set to be at $t=0$ by a translation). From Lemma \ref{lemma:causal-char}, the causal character of the matching hypersurface $\Sigma=\{r=r_0\}$ is determined by the sign of $\alpha(t)$: $\Sigma$ is timelike (null, spacelike) iff $\alpha$ is positive (zero, negative). For this determination, we can use 
\be \alpha = 1-r_0^2J_\pm^2(t).\label{alpha-this-case} \ee Then some straightforward calculations show that the properties of the following proposition hold. Recall that the weak energy condition is given by (\ref{eq:wec}).}

{\begin{proposition}\label{prop:sol-nonnull-r-noncon}
     Let $(M^+,g^+$) be a McVittie spacetime and let $(M^-,g^-)$ be a spatially flat FLRW spacetime, and let these spacetimes be matched, in the sense of Definition IV.1, across a non-null hypersurface $\Sigma$ along which $r=r_0$ is constant. Then the Hubble functions of the McVittie and FLRW regions are given by (\ref{H-nonnull-rcon-sol}) and (\ref{J-nonnull-rcon-sol}) respectively. The time function $\tau$ of the FLRW region can be set equal to the time coordinate $t$ of the McVittie region: $\tau=t$. Four possibilities arise, the first two of which have the following features: 
     \begin{itemize}
     \item[(i)] $H=H_+(t), J=J_+(t), t\in (0,+\infty)$. 
     \\
     The Hubble functions satisfy $H_+(t)>0$ and $J_+(t)<0$ for all $t\in(0,+\infty)$ and have the limiting behaviours
     \be H_+(t) \sim \frac{(1+f_0^{1/2})r_0}{3m}t^{-1},\quad  J_+(t) \sim -\frac{(1+f_0^{1/2})r_0}{3m}t^{-1},\qquad t\to 0^+;\label{HJ-limit-plus-pos-inf} \ee
     \be \lim_{t\to+\infty} H_+(t) = \frac{f_0^{1/2}}{r_0},\quad \lim_{t\to+\infty} J_+(t) = -\frac{1}{r_0}.\label{HJ-limit-plus-pos-zero} 
     \ee
      The scale factor of the McVittie region vanishes in the limit $t\to0^+$, while the scale factor of the FLRW region diverges in this limit. The energy density of both spacetimes is infinite in the limit $t\to 0^+$. Taking the different values for the cosmological constant $\Lambda_{\rm{McV}} = 3f_0/r_0^2$ and $\Lambda_{\rm{FLRW}}=3/r_0^2$ in the different regions, the energy density vanishes in the limit $t\to+\infty$. The weak energy condition is satisfed in the McVittie region and violated in the FLRW region for all $t\in (0,+\infty)$. The matching hypersurface is spacelike everywhere. 
       \item[(ii)] $H=H_-(t), J=J_-(t), t\in(0,+\infty)$.
     \\
    Define $t_1 = t_0\arccosh(f_0^{-1/2})>0$. Then $H_-(t)<0$ on $(0,t_1)$ and $H_-(t)>0$ on $(t_1,+\infty)$. $J_-(t)<0$ for all $t\in(0,+\infty)$. The limiting behaviours at the endpoints are 
    \be H_-(t) \sim \frac{(f_0^{1/2}-1)r_0}{3m}t^{-1},\quad  J_-(t) \sim \frac{(f_0^{1/2}-1)r_0}{3m}t^{-1},\qquad t\to 0^+;\label{HJ-limit-minus-pos-zero} \ee
     \be \lim_{t\to+\infty} H_-(t) = \frac{f_0^{1/2}}{r_0},\quad \lim_{t\to+\infty} J_-(t) = -\frac{1}{r_0}.\label{HJ-limit-minus-pos-inf} 
     \ee
     The scale factor in both regions diverges in the limit as $t\to 0^+$, and the energy density in both regions is infinite in this limit. Taking the different values for the cosmological constant $\Lambda_{\rm{McV}} = 3f_0/r_0^2$ and $\Lambda_{\rm{FLRW}}=3/r_0^2$ in the different regions, the energy density vanishes in the limit $t\to+\infty$. The weak energy condition is violated for all $t\in(0,+\infty)$ in the McVittie region, and is violated for $t\in(0,t_1)$ but satisfied for $t>t_1$ in the FLRW region. There exists $t_2\in(0,t_1)$ such that $\Sigma$ is spacelike for $t\in(0,t_2)$ and timelike for $t\in(t_2,+\infty)$. 
    \end{itemize}
    The corresponding properties of the Case (iii) $H=H_+(t), J=J_+(t), t\in(-\infty,0)$ and Case (iv) $H=H_-(t), J=J_-(t), t\in(-\infty,0)$ can be deduced from the properties of Cases (i) and (ii) by using the fact that $H_\pm,J_\pm$ are odd functions of $t$.\hfill$\blacksquare$
\end{proposition}}

{We see that in every case, the matched spacetime has some pathological features. There are gross violations of the weak energy condition on a scale such that both the scale factor and the density diverge at the singularity. This corresponds to both infinite density \textit{and} infinite volume. So while a matching hypersurface and corresponding McVittie and FLRW spacetimes meeting the criteria of the section exist, they have highly unphysical features.}

\subsection{Non-null hypersurfaces of non-constant radius.}\label{subsec:sol-r-noncon}

In this case, $\Sigma$ may change its causal character, but does so only at isolated points. This includes the important case of a timelike matching hypersurface in which the FLRW region acts as a spatially bound source for the exterior McVittie spacetime (or vice versa). 

The matching conditions are the preliminary matching conditions $r^+=r^-$ and \begin{equation}\del^+=\del^-, \qquad \del = \alpha\dot{t}^2+2\beta\dot{r}\dot{t}-\gamma\dot{r}^2,\label{eq:main-pre1} \end{equation}
along with the junction conditions of Section \ref{subsubsec:non-null-r-noncon}: 
\begin{equation}\Gamma^+=\Gamma^-, \qquad \Gamma = \alpha\dot{t}+\beta\dot{r},\label{eq:main-jct1} \end{equation}and 
\begin{equation}\Pi^+=\Pi^-, \qquad \Pi=\gamma^{-1/2}H'\dot{t}(\beta\dot{t}-\gamma\dot{r}).\label{eq:main-jct2} \end{equation}
The mass continuity equation 
\begin{equation}\alpha^+=\alpha^-, \qquad \alpha = 1-\frac{2m}{r}-r^2H^2 \label{eq:main-mass}\end{equation}
follows from the identity $\Gamma^2=\alpha\del+\dot{r}^2$ and the non-null condition $\Delta\neq0$. 

We consider these matching conditions from the following perspective: Given a McVittie spacetime characterized by the mass parameter $m>0$ and Hubble function $t\mapsto H(t)$, can we find an FLRW spacetime (characterized by Hubble function $\tau\mapsto J(\tau)$) and a matching hypersurface $\Sigma$ that fulfill the matching conditions of Definition \ref{def:match}? 
To answer this question in the affirmative, we must find 
the functions $r(\lambda),t(\lambda)$ characterizing $\Sigma$ as a hypersurface embedded in the McVittie spacetime, the Hubble function  $\cj(\tau)$ of the FLRW spacetime and the function $\tau(\lambda)$ that (with $r(\lambda)$) characterizes $\Sigma$ as a hypersurface embedded in the FLRW spacetime.

Evidently, there is a problem with this. We are seeking four functions $(r(\lambda),t(\lambda),\tau(\lambda))$ and $\cj(\tau)$, but we have only three equations: the first order ODEs (\ref{eq:main-pre1})-(\ref{eq:main-jct2}). We resolve this by exploiting the coordinate freedom on $\sigma$. The matching conditions are invariant under coordinate transformations (diffeomorphisms) on $\sigma$ of the form $\lambda\mapsto \bar{\lambda}(\lambda)$. We can exploit this to reduce the effective degrees of freedom in (\ref{eq:main-pre1})-(\ref{eq:main-jct2}). A convenient choice is to take 
\be \lambda = t, \label{gauge-t-lam} \ee
so that $\dot{t}\equiv 1$ on $\Sigma$. The tangential derivative $\partial_\lambda$ on $\Sigma$ is then $\partial_t$. We will assume this condition henceforth. We find that it allows us to prove the existence of semi-global solutions that satisfy this condition.

There are two cases to consider, depending on whether or not the matching is along a horizon - that is, a hypersurface foliated by marginally trapped surfaces. As seen in Section \ref{subsec:conseq-trapping}, this occurs if $\alpha=0$ along $\Sigma$. 

\subsubsection{Matching across a horizon: $\left.\alpha\right|_\Sigma=0$.}\label{subsubsec:hor}

We assume $\Delta\neq0$ and $\dot{r}\neq0$. Then as we have seen, the matching conditions imply that $\alpha^+=\alpha^-$. With $\alpha^+=\alpha^-=0$ and $\dot{r}^+=\dot{r}^-\neq 0$, the matching condition (\ref{eq:main-jct1}) yields 
\begin{equation}\beta^+ = \beta^-,\qquad \beta = rH\gamma^{1/2}. \label{eq:main-beta-cts} \end{equation}
With $\alpha=0$, (\ref{al-be-ga-id}) yields $\beta^2=1$, and so there are two subcases to consider. 

Consider first the case $\beta^+=\beta^-=1$. Using $\del^+=\del^-$, we can write 
\begin{equation}\dot{\tau} = 1+\left(\frac{1-\gamma}{2}\right)\dot{r}. \label{eq:al0bep1-tau}\end{equation}
From $\beta^-=1$, we have
\begin{equation}\cj = \frac{1}{r},\label{eq:al0bep1-J}\end{equation}
and taking the tangential derivative gives 
\begin{equation}\cj'\dot{\tau} = -\frac{\dot{r}}{r^2}.\label{eq:al0bep1-Jdot}\end{equation}
This equation and (\ref{eq:al0bep1-tau}) allow us to evaluate $\Pi^-$ in terms of $r$ and $\dot{r}$. Similarly, we can evaluate $\Pi^+$ in terms of $r$ and $\dot{r}$ by taking the tangential derivative of 
\begin{equation}H= \frac{1}{r\gamma^{1/2}}=\frac{1}{r}\left(1-\frac{2m}{r}\right)^{1/2},\label{eq:al0bep1-H}\end{equation}
which is equivalent to $\beta^+=1$. We find that $\Pi^+=\Pi^-$ is equivalent to 
\begin{equation}\dot{r}=\frac32\left(1-\frac{2m}{r}\right). \label{eq:al0bep1-rdot} \end{equation}
This has the implicit general solution 
\begin{equation}t = \frac23\left(r+2m\log\left(\frac{r}{2m}-1\right)\right) + c_1, \label{eq:eq:al0bep1-r-sol}\end{equation}
which can be (formally) inverted to give $r(t)$ in terms of the product log function and the constant of integration $c_1$. Then (\ref{eq:al0bep1-H}) gives the Hubble function  $H$ of the McVittie region as a function of the global time coordinate $t$ of the McVittie region. Combining (\ref{eq:al0bep1-rdot}) and (\ref{eq:al0bep1-tau}) and integrating yields 
\begin{equation}\tau=\frac23\left(r+3m\log\left(\frac{r}{2m}-1\right)\right) + c_2. \label{eq:eq:al0bep1-r-sol-tau}\end{equation}
Inverting and using (\ref{eq:al0bep1-J}) gives the Hubble function  $J$ of the FLRW region as a function of the global time coordinate $\tau$ of the FLRW region. 
We note that 
\begin{equation}\del = -\frac94\left(1-\frac{2m}{r}\right) \label{eq:al0bep1-del} \end{equation}
along the matching hypersurface, which is therefore timelike everywhere (see (\ref{eq:normsmcv})). Note that (\ref{eq:al0bep1-del}) and (\ref{eq:al0bep1-rdot}) show that the assumptions of this case are indeed satisfied by the solutions derived. 

As seen previously, an additional constraint on $\Sigma$ - in this case, that it is a horizon of the ambient spacetimes - imposes strong restrictions on those spacetimes. Given the mass parameter $m$, there is just a one parameter family of (expanding) McVittie spacetimes which allow the matching scenario of this subsection. 

In the case $\beta^\pm=-1$, we find corresponding restrictions, and the configurations that arise are precisely the time-reverse of those arising from $\beta^\pm=1$. 

\subsubsection{Matching across a non-horizon: $\left.\alpha\right|_\Sigma\neq0$.}\label{subsubsec:general}

We arrive finally at what should be considered the most general case of matching. The hypersurface is non-null except at isolated points, the radius $r$ is non-constant on $\Sigma$ (but with $\dot{r}$ possibly vanishing at isolated points) and $\alpha\neq 0$ on $\Sigma$ - again, except possibly at isolated points. 

Since neither $\del$ nor $\alpha$ vanish on $\Sigma$, the identity $\Gamma^2=\alpha\del+\dot{r}^2$ shows that (\ref{eq:main-pre1}) and (\ref{eq:main-jct1}) are equivalent to (\ref{eq:main-mass}) and (\ref{eq:main-jct1}). Thus we can state the following: 

\begin{proposition}\label{prop:necess-suffice} The necessary and sufficient conditions for continuously matching a McVittie spacetime described by (\ref{mcv-lel})-(\ref{gam-def}) to an FLRW spacetime described by (\ref{rw-lel}) across a hypersurface $\Sigma$ which is non-null (except at isolated points), on which $\dot{r}$ is non-vanishing (except at isolated points) and which is not foliated by marginally trapped surfaces is that there exist an interval $I$, Hubble functions $H(t), \cj(\tau)$ and solutions $(r^+(\lambda)=r^-(\lambda),t(\lambda),\tau(\lambda))$ on $I$ of the equations (\ref{eq:main-jct1}), (\ref{eq:main-jct2}) and (\ref{eq:main-mass}), with the maximal rank conditions (\ref{max-rank}) satisfied almost everywhere on $I$ and with $\Delta, \alpha$ and $\dot{r}$ vanishing {at most} only at isolated points of $I$. 
\hfill$\blacksquare$
\end{proposition}

\begin{comments} This proposition shows that the matching problem (i.e.\ that of finding, for a given McVittie spacetime, a corresponding FLRW spacetime and matching hypersurfaces $\Sigma^\pm$ satisfying the conditions of Definition \ref{def:match}) can be cast in terms of a system of ODEs subject to an algebraic constraint. Considering the local and global existence and uniqueness properties of this system allows us to solve the matching problem. We establish the results below by imposing the gauge (\ref{gauge-t-lam}), and showing that this leads to semi-global results for the existence of $\cj$ and $\Sigma^\pm$. Global existence to the future is established, and we show that the matching must terminate at the past singularity of the McVittie region. 
\end{comments}

We set $\lambda=t$ and define
\begin{equation} J(t) = \cj(\tau(t)). \label{J-def} \end{equation}
Then
\begin{equation} \jd = \cj'\dot{\tau}. \label{eq:J-dot} \end{equation} 
We do not need to distinguish $r^+$ and $r^-$ {(recall that the preliminary matching conditions yield $[r]=0$, i.e.\ $r^+=r^-$. Furthermore, the assumed smoothness properties of the embeddings allow us to take tangential derivatives to get $\dot{r}^+=\dot{r}^-$).} Using the definitions above, the matching conditions (\ref{eq:main-jct1}), (\ref{eq:main-jct2}) and (\ref{eq:main-mass}) can be written as
\begin{equation} \alpha + \beta\dot{r} = (1-r^2J^2)\dot{\tau}+rJ\dot{r},\label{eq:match1}  
\end{equation} 
\begin{equation}
\gamma^{-1/2}H'(\beta-\gamma\dot{r}) = \dot{J}(rJ\dot{\tau}-\dot{r}),\label{eq:match2}
\end{equation} 
\begin{equation} 
\frac{2m}{r^3} + H^2 = J^2. \label{eq:match3} 
\end{equation} 
(Here and in the remainder of the paper, $\alpha,\beta$ and $\gamma$ will refer to the McVittie forms of (\ref{alpha-def})-(\ref{gam-def}).)

Evidently, we could solve (\ref{eq:match3}) for $r$ and eliminate this variable from the system. However, it is more convenient to proceed as follows. Taking a tangential derivative, we see that (\ref{eq:match3}) is equivalent to the initial value problem 
\begin{equation} \frac{3m}{r^4}\dot{r} + J\dot{J} = HH', \label{eq:match4a} \end{equation} 
\begin{equation} \frac{2m}{r_0^3}+H_0^2=J_0^2,
\label{eq:match4b} 
\end{equation}
where $H_0, J_0$ and $r_0$ are initial values of $H,J$ and $r$ respectively at some initial time $t_0$. $H_0=H(t_0)$ is determined by the given McVittie spacetime; $r_0$ and $J_0$ are then specified subject to (\ref{eq:match4b}). 

The matching conditions then become a system of ODEs for the triple $(r,\tau,J)$ defined on an interval $I$ with $t_0\in I$, along with the initial condition (\ref{eq:match4b}). The Robertson-Walker spacetime is constructed by determining the Hubble function
\begin{equation} \cj(\tau) = J(t(\tau)) \label{eq:J-tau-sol} \end{equation}
where $\tau\mapsto t(\tau)$ is the inverse of $t\mapsto \tau(t)$ on the range of $I$ under $\tau$. We note that the matching condition (\ref{eq:match3}) - or equivalently the IVP (\ref{eq:match4a}), (\ref{eq:match4b}) - is the equation $\alpha=1-r^2J^2$. We can therefore use these terms interchangeably throughout the analysis of the system of ODEs: any such interchange yields an equivalent system. In particular, and recalling that $\alpha$ is assumed non-zero, we can write (\ref{eq:match1}) as 
\be \dot{\tau} = \left(\frac{\beta-rJ}{1-r^2J^2}\right)\dot{r}+1.\label{eq:tau-dot0} \end{equation}
 We use this form to eliminate $\dot{\tau}$ from (\ref{eq:match2}), and we can then combine the resulting equation with (\ref{eq:match4a}) to solve for $\dot{r}$ and $\dot{J}$. Two solutions emerge, the first of which is $\dot{r}=0, J\dot{J}=HH'$. This yields $m=0$, and so a trivial configuration. The second solution yields (where we also substitute into (\ref{eq:tau-dot0}) to obtain a diagonalized system)
\begin{eqnarray}
\dot{\tau}&=& \left(\frac{1-r^2J^2}{1-rJ\beta}\right)+\frac{r^3}{3m}\frac{(\beta-rJ)^2}{(1-rJ\beta)(1-r^2J^2)}\gamma^{-1/2}H',\label{eq:tau-dot1}\\
    \dot{r}&=&rJ\left(\frac{1-r^2J^2}{1-rJ\beta}\right)+\frac{r^3}{3m}\left(\frac{\beta-rJ}{1-rJ\beta}\right)\gamma^{-1/2}H',\label{eq:rdot-f} \\
    \dot{J}&=&\left(\frac{1-r^2J^2}{1-rJ\beta}\right)\left(-\frac{3m}{r^3}+\gamma^{1/2}H'\right).\label{eq:jdot-final}
\end{eqnarray}

With the matching conditions reduced to this system of ODEs and the initial condition (\ref{eq:match4b}), we find that we can indeed solve the matching problem (i.e.\ given the McVittie spacetime, we can find an FLRW spacetime and embedded hypersurfaces $\Sigma^\pm$ that satisfy the matching conditions of Definition \ref{def:match}). This is the main result of this paper, and is encapsulated in {Theorem \ref{Thm:Match-V4}} below. 

We need to establish some properties of the coefficients of (\ref{eq:tau-dot1})-(\ref{eq:jdot-final}) which we do here before stating and proving the main result. 

\begin{lemma}\label{lemma:bounds}
    Let 
    \be a= \frac{1-r^2J^2}{1-rJ\beta},\quad b =\frac{\beta-rJ}{1-rJ\beta},\quad c = \frac{\beta-rJ}{1-r^2J^2} \label{abc-def} \ee
    and define 
    \be x = rJ,\quad y=rH>0,\quad f=\gamma^{-1}=1-\frac{2m}{r}\in(0,1).  \label{xy-def} \ee 
    Assume that the matching condition (\ref{eq:match3}) holds, so that 
    \be x^2=y^2 + 1-f. \label{x-y-f-id} \ee
    \begin{itemize}
        \item[(i)] If $x>0$, then $\alpha=0$ iff $(x,y)=(1,f^{1/2})$. If $x\neq 1$, we can write 
        \be a = \frac{1+xyf^{-1/2}}{1+x^2f^{-1}},\quad b= \frac{yf^{1/2}-x}{1+y^2},\quad c = \frac{1-f^{-1}}{yf^{-1/2}+x}. \label{abc-pos} \ee
        The following bounds hold, along with the limiting values at $(x,y)=(1,f^{1/2})$ indicated: 
    \be f<a<f^{1/2},\quad \left.a\right|_{\alpha=0}= \frac{2f}{f+1}, \label{a-props-pos} \ee
    \be \frac{f^{1/2}-1}{y} < b <\frac{f^{1/2}-1}{x},\quad \left.b\right|_{\alpha=0} = \frac{f-1}{f+1}, \label{b-props-pos} \ee
    and 
    \be \frac{1-f^{-1/2}}{y} < c <\frac{1-f^{-1/2}}{x},\quad \left.c\right|_{\alpha=0} = \frac{f-1}{2f}. \label{c-props-pos} \ee
    In particular, if $x>0$ then $a>0, b<0$ and $c<0$.
    \item[(ii)] If $x<0$, then $\alpha=0$ iff $(x,y)=(-1,f^{1/2})$. If $x\neq-1$, we can write 
    \be a=\frac{1-x^2}{1-xyf^{-1/2}},\quad b = \frac{yf^{-1/2}-x}{1-xyf^{-1/2}},\quad c = \frac{yf^{-1/2}-x}{1-x^2}. \label{abc-neg} \ee
    The following bounds and limiting values hold:
    \be -f^{1/2} < a < f^{1/2},\quad \left.a\right|_{\alpha=0} = 0 \label{a-props-neg} \ee
    and 
    \be \frac{1-f^{1/2}}{(-x)}<b<(-x)(1+f^{-1/2}),\quad \left.b\right|_{\alpha=0}=1,\label{b-props-neg} \ee
    In particular, if $x<0$ then $b>0$ and $a$ and $c$ have the same sign as $\alpha$, with 
    \be \lim_{\alpha\to 0^\pm} c = \pm \infty. \label{c-lim-neg} 
    \ee
    
    \end{itemize}
\end{lemma}

\noindent\textbf{Proof:} These rely on relatively straightforward algebraic manipulations. To obtain the first of (\ref{abc-pos}) and the bounds of (\ref{a-props-pos}), we use 
\be a = \frac{1-x^2}{1-xyf^{-1/2}} = \frac{(1-x^2)(1+xyf^{-1/2})}{1-x^2y^2f^{-1}}=\frac{(1-x^2)(1+xyf^{-1/2})}{(1-x^2)(1+x^2f^{-1})}.\ee
To obtain the second equality, we multiply by 1 (in a form non-zero/non-zero) and to obtain the third, we use (\ref{x-y-f-id}). For $x^2\neq1$, the obvious cancellation yields a form of $a$ which allows us to establish the bounds in a straightforward way. Taking the limit $(x,y)\to(1,f^{1/2})$ gives the value of $a$ at $\alpha=0$ when $x>0$. The other statements are proven in a similar way. \hfill$\blacksquare$

We can now state and prove our main results. We separate these into statements about the system of equations (\ref{eq:tau-dot1})-(\ref{eq:jdot-final}), and consequent statements about the spacetime matching problem.

\begin{theorem}\textbf{Existence, uniqueness and properties of matching configurations.}\label{Thm:Match-V4} 
Let $m>0$ and let $H\in C^2((0,+\infty),\mathbb{R})$ with the following properties:
\be H(t)>0,\quad H'(t)<0\quad \hbox{for all } t\in(0,+\infty), \label{H-exp-wec} \ee
\be \lim_{t\to 0^+} H(t) = +\infty, \label{H-bb} \ee
and 
\be \lim_{t\to +\infty} \int_{t_0}^t H(s) ds = +\infty. \label{H-scale-inf} \ee
The limit $H_\infty=\lim_{t\to+\infty}H(t)$ exists and is non-negative. In the case where $H_\infty=0$, assume that 
\be H' \sim - \kappa H^2,\quad t\to+\infty \label{dH-lim-zero} \ee
for some $\kappa>0$. Let $t_0>0$, $H_0=H(t_0)$, let $J_0\in\real$ and $r_0$ satisfy (\ref{eq:match4b}) with $r_0>2m$, and let $\tau_0\in\real$. 

\begin{itemize}
    \item[(i)] Let $J_0>0$. Then there exists $t_1\in(0,t_0)$  and a unique solution on $I=(t_1,+\infty)$ (the maximal interval of existence) of the initial value problem comprising the ODEs (\ref{eq:match1}), (\ref{eq:match2}) and (\ref{eq:match4a}) and the initial conditions
\be J(t_0) = J_0,\quad r(t_0)=r_0,\quad \tau(t_0)=\tau_0. \label{eq:ICS} \end{equation}
$J>0$ throughout $I$, and the following limits hold: 
\be
\lim_{t\to t_1^+} r(t) = 2m \label{r-to-2m} 
\ee
and
\be
\lim_{t\to+\infty} r(t) = +\infty. \label{r-at-infinity} 
\ee
These conclusions hold for any $\alpha_0\neq0$. If $\alpha_0=0$, the conclusions also hold provided \be H(t_0)H'(t_0)+\frac{3f(r_0)}{2r_0^2}\left(1-\frac{3m}{r_0}\right) \neq 0.\label{H-not-special} \ee

\item[(ii)] Let $J_0>0$ and let $\tau:(t_1,+\infty)$ be the solution for $\tau$ of (\ref{eq:match1}), (\ref{eq:match2}) and (\ref{eq:match4a}) obtained in part (i). Then there exists $t_2\in(t_1,t_0)$ such that $\dot{\tau}(t)$ is negative for all $t\in(t_1,t_2)$. Further, $\dot{\tau}$ is positive for all sufficiently large $t$. $\dot{\tau}$ remains finite throughout $(t_1,+\infty)$.   

\item[(iii)] Let $J_0<0$. If $\alpha(t_0)>0$, then there exists $t_1<t_0$ such that $\alpha(t_1)<0$ for $t<t_1$. Then matching is ruled out by the observations of Section \ref{subsec:conseq-trapping}, and this conclusion also holds if $\alpha(t_0)<0$.   

\end{itemize}
\end{theorem}

\begin{comments}
\begin{itemize} 
\item[(i)] This proposition validates the gauge choice (\ref{gauge-t-lam}). It shows that with this choice, we can always solve the matching problem - at least locally.
\item[(ii)] The condition (\ref{H-not-special}) must be imposed in order to avoid the solution of Section \ref{subsubsec:hor}, for which the term on the left equals zero. As we will see, the condition allows us to determine uniqueness in the case $\alpha_0=0$.
\item[(iii)] The conditions on the Hubble function of the McVittie spacetime characterize an expanding McVittie spacetime, in which the \textit{background} FLRW spacetime (i.e.\ the spacetime obtained by setting $m=0$ in (\ref{mcv-lel}), which we emphasize is not the FLRW spacetime of the matching) satisfies the weak energy condition. {In the present case, $H'(t)<0$ follows from the weak energy condition (\ref{eq:wec}).} There is a big bang singularity at a finite time in the past, which we set to be at $t=0$ by a translation, and the scale factor (the exponential of the integral of $H$) becomes infinite as $t\to+\infty$. These conditions apply to the $\Lambda$CDM model, and to models with a linear barotropic equation of state. See e.g.\ the comments in Section 1.1 of \cite{nolan2025can}. 
\item[(iv)] To understand (\ref{dH-lim-zero}), we note the following. The density $\mu_0$ and pressure $P_0$ of the FLRW background of the McVittie region (i.e.\ the $m=0$ limit of (\ref{mcv-matter})) are given by 
\be 8\pi\mu_0 = 3H^2-\Lambda,\quad 8\pi P_0 = -2H'-3H^2+\Lambda. \label{eq:FLRW}\ee
Suppose that there is a barotropic equation of state of the form 
\be P_0=g(\mu_0),\quad g:\real_+\to\real\label{eq:eos} \ee where $g$ is differentiable at the origin. We introduce the parameter 
\be \kappa = \frac32(1+g'(0))>0, \label{eq:kappa-def}\ee 
which relates to the sound speed at zero density. The sign follows from the weak energy condition. Then in the case where $H_\infty=\lim_{t\to +\infty} H(t) = 0$, we can prove (\ref{dH-lim-zero}). See Lemma 3.3 of \cite{nolan2025can}. This condition is used in the proof of part (ii) of the theorem. 
\end{itemize} 
\end{comments}

\noindent\textbf{Proof of Theorem \ref{Thm:Match-V4} (i):} To prove part (i) of the theorem, we establish first local existence, and then show that the maximal interval of existence has the form claimed. We do this by replacing the system of equations (\ref{eq:tau-dot1})-(\ref{eq:jdot-final}) with an auxiliary system, motivated by Lemma \ref{lemma:bounds}. As seen, the system (\ref{eq:tau-dot1})-(\ref{eq:jdot-final}) follows from the system of interest (\ref{eq:match1}), (\ref{eq:match2}) and (\ref{eq:match4a}) in the case when $\alpha\neq 0$. We work with this auxiliary system, and then show that solutions of this system correspond to solutions of the original system. So with $x,y$ as in Lemma \ref{lemma:bounds} (but \textit{not} assumed to satisfy the matching condition (\ref{x-y-f-id})), define 
\be \mathcal{A} = \frac{1+xyf^{-1/2}}{1+x^2f^{-1}},\quad \mathcal{B}= \frac{yf^{1/2}-x}{1+y^2},\quad \mathcal{C} = \frac{1-f^{-1}}{yf^{-1/2}+x}, \label{abc-bar} \ee
and consider the initial value problem consisting of the ODEs
\begin{eqnarray}
    \dot{\tau} &=& \mathcal{A}+\frac{r^3}{3m}\mathcal{B}{\mathcal{C}}\gamma^{-1/2}H',\label{dtau-bar} \\
    \dot{r} &=& \mathcal{A}rJ +\frac{r^3}{3m}\mathcal{B}\gamma^{-1/2}H',\label{dr-bar} \\
    \dot{J} &=&\mathcal{A}\left(-\frac{3m}{r^3}+\gamma^{1/2}H'\right) \label{dJ-bar} 
\end{eqnarray}
with initial values (\ref{eq:ICS}) satisfying (\ref{eq:match4b}) and $J_0>0$. Since $x_0=r_0J_0>0$, $y_0=r_0H_0>0$ and $f\in(0,1)$, it is clear that the right hand sides of (\ref{dtau-bar})-(\ref{dJ-bar}) are $C^1$ on a neighbourhood of $(\tau_0,r_0,J_0)$, and so local existence and uniqueness of solutions is immediate. We note that this holds even if $\alpha_0=0$.  Let $I$ be the corresponding maximal interval of existence. 

A direct calculation shows that (\ref{eq:match4a}) holds on $I$, and so with the initial condition (\ref{eq:match4b}), we see that (\ref{eq:match3}) holds throughout $I$.  It follows that $J$ is non-vanishing, and hence positive throughout $I$. With (\ref{eq:match3}) in hand, we see that the coefficients $\mathcal{A},\mathcal{B}$ and $\mathcal{C}$ satisfy the conclusions for $a,b,c$ of part (i) of Lemma \ref{lemma:bounds}. Then using $H'<0$ and this lemma, (\ref{dr-bar}) and (\ref{dJ-bar}) show that $\dot{r}>0$ and $\dot{J}<0$ on $I$. 

We consider next the future evolution. So let $t_0<t\in I$. Using the bounds (\ref{a-props-pos}) and (\ref{b-props-pos}), we can show 
\be \dot{r} < rf^{1/2}J-\frac{r^2}{3m}f^{1/2}(1-f^{1/2})\frac{H'}{H}. 
\label{rdot-bound1} 
\ee
Recall that $H'<0$. Define 
\be u(r) = rf^{1/2}(1-f^{1/2}). \label{u-def} \ee
Then $u$ is monotone increasing on $(2m,+\infty)$, and $\lim_{r\to+\infty} u(r) = m$. Thus $u(r_0)<u(r(t))<m$ for all $t>t_0$. Using this upper bound and the fact that $J$ is decreasing, we obtain 
\be \dot{r} < r(J_0 - \frac{1}{3}\frac{H'}{H}). \label{rdot-bound2} \ee
Integrating yields 
\be r(t) < r(t_0)\left(\frac{H_0}{H(t)}\right)^{1/3}e^{J_0(t-t_0)},\quad t>t_0. \label{eq:r-bound} \ee
Since $J$ is decreasing, we have $J_0>J(t)$ for $t>t_0$. With $J^2=H^2+2m/r^3$, we have 
\be J > H >H_{\infty} \label{J-lower} \ee
where 
\be H_\infty = \lim_{t\to\infty} H \geq0.\label{H-inf} \ee
This limit must exist and be non-negative since $H$ is positive and decreasing. These bounds of $r$ and $J$ show that the system comprising (\ref{dr-bar}) and (\ref{dJ-bar}) has solutions which remain finite for all finite $t>t_0$. Along with the bounds of part (i) of Lemma \ref{lemma:bounds}, this is sufficient to ensure that the solution extends to all $t>t_0$. Since the solution is $C^1$, the right hand side of (\ref{dtau-bar}) is continuous, yielding existence and uniqueness of the solution for $\tau$ and hence for the full system on $[t_0,+\infty)$. 

For the past evolution, we use the lower bounds of (\ref{a-props-pos}) and (\ref{b-props-pos}) to write 
\be \dot{r} > rfJ-\frac{r^2}{3m}f^{1/2}(1-f^{1/2})\frac{H'}{J}.\label{rdot-above} \ee
Since $r>2m$ throughout the McVittie region, this bound also holds on $\Sigma$. Then 
\be J^2 = H^2 +\frac{2m}{r^3} < H^2+\frac{1}{4m^2},\label{eq:J-bound} \ee
which yields a lower bound for the positive term $1/J$. We also have the upper bound $r<r_0$ for $t<t_0$. Then defining 
\be v = 1-f^{1/2},\label{eq:v-def} \ee
(\ref{rdot-above}) yields 
\be \dot{v} <\left(\frac{m}{r_0}J_0 + \frac{H'}{3(H^2+k^2)^{1/2}}\right)v,\quad k=1/2m. \label{eq:vdot} \ee
Integrating over $(t,t_0)$ gives 
\be v(t) > C_0\exp\left(-\frac{m}{r_0}J_0(t_0-t)\right)\left(\frac{\sqrt{H^2+k^2}+H}{\sqrt{H^2+k^2}-H}\right)^{1/6},\quad t<t_0, \label{v-lower-bound} \ee
where $C_0$ is a constant. 
Since $H\to+\infty$ as $t\to 0^+$, this shows that $v=1$ at some time $t_1\in(0,t_0)$: this is equivalent to $r\to 2m$ as $t\to t_1^+$. This limit and the bounds of part (i) of Lemma \ref{lemma:bounds} show that the right hand sides of (\ref{dtau-bar}) and (\ref{dJ-bar}) are bounded on $(t_1,t_0]$ and so the solutions for $\tau$ and $J$ exist on this interval (and so on $(t_1,+\infty)$). 

We note also that 
\be H\to H_1 = H(t_1),\quad J\to J_1 = \left(H_1^2+\frac{1}{4m^2}\right)^{1/2}\quad \hbox{as} \quad t\to t_1^+. \label{h-j-past-limits} \ee
and \be \lim_{t\to t_1^+} \dot{r}=0. \label{rdot-limit} \ee
This last limit follows from (\ref{dr-bar}) and the bounds (\ref{a-props-pos}) and (\ref{b-props-pos}). 

To prove (\ref{r-at-infinity}), we note that both terms on the right hand side of (\ref{dr-bar}) are positive. We can drop the second term (with $H'$) and apply the lower bound on $\mathcal{A}$ from (\ref{a-props-pos}) to obtain
\be \dot{r} > (r-2m)H. \label{eq:rdot-inf} \ee
The result follows immediately by integrating and using (\ref{H-scale-inf}). 

To complete the proof of part (i), we verify by a direct calculation that the solutions of (\ref{dtau-bar})-(\ref{dJ-bar}) are also solutions of (\ref{eq:match1}), (\ref{eq:match2}) and (\ref{eq:match4a}). Uniqueness of solutions of the latter system follows immediately if $\alpha_0\neq0$, as this allows us to write (\ref{eq:match1}), (\ref{eq:match2}) and (\ref{eq:match4a}) in the form (\ref{dtau-bar})-(\ref{dJ-bar}). If $\alpha_0=0$, we obtain uniqueness as follows: a subtlety arises because (\ref{eq:match1}) is empty at $t=t_0$ in this case, and the system is underdetermined. {(To clarify, in this case we have $\alpha=1-r^2\cj^2=0$ at $t=t_0$. Then $\beta^+=\beta^-=1$ at $t_0$, using (\ref{al-be-ga-id}) and the co-expanding conditions $H(t_0)>0,J(t_0)>0$. The coefficients of $\dot{r}$ on both sides of (\ref{eq:match1}) are therefore equal.)} We know that the relevant IVP has at least one solution, namely the solution $(\tau,r,J)_{\textrm{aux}}$ of the auxiliary IVP comprising (\ref{dtau-bar})-(\ref{dJ-bar}) and the initial values (\ref{eq:ICS}). These initial values are subject to 
\be \alpha_0 = 1-\frac{2m}{r_0}-r_0^2H_0^2 = 1-r_0^2J_0^2 = 0. \label{ICS-alpha-zero} \ee
Let $(\tau,r,J)$ be any other solution of the IVP. Then $\alpha=1-r^2J^2$ and $J>0$. Using (\ref{eq:match2}) and (\ref{eq:match4a}), we calculate 
\be \left.\frac{d(rJ)}{dt}\right|_{t=t_0}=\left.\frac{2f}{r(f+1)}\left(1-\frac{3m}{r}\right) +\frac{4r^2}{3(f+1)}HH'\right|_{t=t_0}. \label{generic-id} \ee
By hypothesis - see (\ref{H-not-special}) - this term is non-zero (but vanishes if $H$ has the form obtained in Section \ref{subsubsec:hor} for the solution with $\alpha\equiv 0$). It follows that $rJ$ is not equal to its initial value of 1 on a punctured neighbourhood $I'$ of $\{t_0\}$, and hence that $\alpha\neq0$ on $I'$. Hence the solution $(\tau,r,J)$ satisfies the auxiliary system (\ref{dtau-bar})-(\ref{dJ-bar}) on $I'$. By continuity, the solution must also satisfy the auxiliary system on $I'\cup\{t_0\}$, and also satisfies the initial conditions (\ref{eq:ICS}). But $(\tau,r,J)_{\textrm{aux}}$ is the unique solution on $I'\cup\{t_0\}$ of this IVP, and uniqueness follows. This completes the proof of part (i) of the proposition.   \hfill$\blacksquare$
\vskip20pt

\noindent\textbf{Proof of Theorem \ref{Thm:Match-V4} (ii):}  We can use the equation (\ref{dtau-bar}) to describe the evolution of $\tau$ on the maximal interval of existence $I$. As $t\to t_1^+$, the proof of {part (i)} shows that $r\to 2m, H\to H_1:=H(t_1)>0$ and $J\to J_1 = (H_1^2+1/4m^2)^{1/2}$. Applying the bounds of Lemma \ref{lemma:bounds}, we see that $\mathcal{A}\to 0$ as $t\to t_1^+$, and that $\mathcal{B}$ and $\mathcal{C}f^{1/2}$ are bounded in the limit, lying in a bounded interval that asymptotes to $(-1/(2mH_1),-1/(2mJ_1))$. Likewise, $H'\to H'(t_1)<0$. It follows that there exists $t_2>t_1$ such that $\dot{\tau}(t)<0$ for all $t\in(t_1,t_2)$. 

For large $t$, there are two cases to consider depending on whether or not $H_\infty=\lim_{t\to+\infty}H(t)\geq0$ is positive or zero. In both cases, $J_\infty=\lim_{t\to+\infty}J(t)=H_\infty$ by the mass continuity equation (\ref{eq:match3}). For $H_\infty>0$, we can apply a squeeze principle and use Lemma \ref{lemma:bounds} to obtain
\be \lim_{t\to+\infty} \mathcal{B} = \lim_{t\to+\infty} \mathcal{C}f^{1/2} = -\frac{m}{H_\infty}r^{-1}, \label{b-c-inf1} \ee
while $\mathcal{A}\to 1$ in this limit. By hypothesis, $H'(t)\to 0$ as $t\to+\infty$. Thus
\be \dot{\tau} \sim 1 +\frac{m}{3H_\infty^2}r^{-1}H' \to 1,\quad t\to+\infty,\label{dtau-pos-inf1} \ee
proving the second part of the result in this case. 

If $H_\infty=0$, we have $\lim_{t\to+\infty} \mathcal{A}=1$ as before, and Lemma \ref{lemma:bounds} yields the lower bound 
\be \frac{r^3}{3m}\mathcal{B}\mathcal{C}\gamma^{-1/2}H' \gtrsim \frac{m}{3r}\frac{H'}{H^2} \sim - \frac{m}{3r}\kappa,\quad t\to+\infty. \label{bound-H0-zero} \ee

That is, the first term is greater than a quantity that asymptotes to the second term, which in turn is asymptotic to the third in the limit $t\to +\infty$. The last relation relies on (\ref{dH-lim-zero}). As in the first case, this yields $\dot{\tau} \sim 1$ as $t\to +\infty$.

Thus $\dot{\tau}$ changes sign and hence must vanish at least once on $I$. 

\hfill $\blacksquare$
\vskip20pt

\noindent\textbf{Proof of Theorem \ref{Thm:Match-V4} (iii):}  For part (iii), we prove the statements in reverse order. We assume the existence of a solution on a maximal interval $I\subset (0,+\infty)$, and show that contradictions arise. As in part (i), the mass continuity condition (\ref{eq:match3}) shows that $J$ is non-zero and hence (in this case) negative throughout $I$. 

If we encounter a point $p\in\Sigma$ at which $\alpha^+=\alpha^-<0$, in particular if $\alpha(t_0)<0$, we see that the matching entails that of an anti-trapped region of the expanding McVittie spacetime with a trapped region of the collapsing FLRW spacetime. As seen in Section \ref{subsec:conseq-trapping}, this is not possible (see also Comment 4 below).

Otherwise, we consider the auxiliary system obtained by replacing the coefficients of (\ref{eq:tau-dot1})-(\ref{eq:jdot-final}) with the $x<0$ forms of Lemma \ref{lemma:bounds}.

Assume that $\alpha>0$ throughout $I$ (we obtain a contradiction). With $x=rJ<0$, the auxiliary system is (\ref{dtau-bar})-(\ref{dJ-bar}) with the coefficients $(\mathcal{A},\mathcal{B},\mathcal{C})$ replaced by the $(a,b,c)$ of (\ref{abc-neg}).  As per Lemma \ref{lemma:bounds}, these coefficients are all positive when $\alpha>0$, and both $r$ and $J$ are decreasing on $I$ while $\alpha$ remains positive. It follows that $y=rH$ is decreasing, and so $y(t)>y(t_0)$ for all $t<t_0, t\in I$. The solution for $(r,J)$ persists for all $t<t_0$ (the relevant coefficients remaining finite and positive) unless we extend to $t\to0^+$. In either case, there must exist $t_*\in(0,t_0)$ at which $y(t_*)=1$ ($y$ cannot have an upper bound less than one as $r\to2m$ and $H$ diverges as $t$ decreases to zero). But then $x^2(t_*)>1$, which corresponds to $\alpha(t_*)<0$. \hfill{$\blacksquare$}

\begin{comments} 
To clarify the proof of the $\alpha<0$ case of part (iii) of Theorem \ref{Thm:Match-V4}, we revisit the discussion of Section \ref{subsec:conseq-trapping} on matching regular, trapped and anti-trapped regions. It is useful for this to consider various vector fields encountered earlier. 

Define the null vector fields 
\begin{eqnarray} \vec{n}^{(1)} &=& (\beta+1)\partial_t-\alpha\partial_r,\label{null-n1} \\
\vec{n}^{(2)} &=& (\beta-1)\partial_t-\alpha\partial_r.\label{null-n2}
\end{eqnarray}
On $\Sigma$, these can be written in the \textit{continuous} form 
\begin{eqnarray}
    \left.\vec{n}^{(1)}\right|_{\Sigma} &=& \frac{\Gamma+\dot{r}}{\sqrt{2}\Delta}\left( (\sqrt{2}-1)\vec{e}_1+\vec{\ell}\right), \label{null-n1-cts} \\
    \left.\vec{n}^{(2)}\right|_{\Sigma} &=& \frac{-\Gamma+\dot{r}}{\sqrt{2}\Delta}\left( (\sqrt{2}+1)\vec{e}_1-\vec{\ell}\right). \label{null-n2-cts}
\end{eqnarray}
The continuity of these forms follows from continuity of the basis vectors, the preliminary matching conditions $\Delta^+=\Delta^-$, $r^+=r^-$ (and the tangential derivative of the latter) and the junction condition $\Gamma^+=\Gamma^-$. 
The time orientation of these null directions is determined by the inner products
\begin{eqnarray}
    g(\vec{u},\vec{n}^{(1)}) = -\gamma^{-1/2}(\beta+1),\label{temp-null1} \\
    g(\vec{u},\vec{n}^{(2)}) = -\gamma^{-1/2}(\beta-1).\label{temp-null2}
\end{eqnarray}
The null directions are future-pointing if and only if these inner products are negative (and so point into the same half of the future causal cone as $\vec{u}$). 

Now consider the situation that holds in part (iii) of Theorem \ref{Thm:Match-V4}. We have $H>0$ and so $\beta^+=rH\gamma^{1/2}>0$ and $J<0$ so that $\beta^-=rJ<0$. When $\alpha<0$, this implies that $\beta^+>1$ and $\beta^-<-1$. Then
\be g(\vec{u},\vec{n}^{(1)})^+<0,\quad g(\vec{u},\vec{n}^{(2)})^+<0,\label{fp-plus} \ee
whereas 
\be g(\vec{u},\vec{n}^{(1)})^->0,\quad g(\vec{u},\vec{n}^{(2)})^->0.\label{fp-minus} \ee
Thus there is a discontinuity in the chronological orientation of the continuous null vector fields $\vec{n}^{(i)}, i=1,2$ across $\Sigma$ in this case, and time orientability of the spacetime is lost. There is a clear solution: to switch the time orientation of (say) the FLRW region. This amounts to swapping the sign of $J$ - and so returning us to parts (i) and (ii) of Theorem \ref{Thm:Match-V4}. 
\end{comments} 

The causal nature of the matching hypersurface depends on the properties of the Hubble function $H$ of the McVittie spacetime, and we can deduce only the following partial results.

\begin{proposition}\label{prop:causal} Let $J_0>0$ and let $\Sigma$ be the matching hypersurface obtained in part (i) of Theorem \ref{Thm:Match-V4}. Then:
\begin{itemize}
   \item[(i)] The matching hypersurface is initially spacelike. {That is, let $t_1<t_0$ be the time at which $\Sigma$ meets $\{r=2m\}$. Then there exists $t_s\in (t_1,t_0)$ such that $\Sigma$ is spacelike on $(t_1,t_s)$.} 
   \item[(ii)] Consider the family of McVittie spacetimes with $\Lambda=3H_\infty^2>0$ and with 
   \be H(t) = H_\infty\coth\left(\frac32(1+\kappa)H_\infty t\right),\quad \kappa \in (-1,1],\quad t>0. \ee 
   This corresponds to the background equation of state $P_0=\kappa\mu_0$. Then $\Sigma$ is timelike at late times for $\kappa>-2/3$ and is spacelike at late times for $\kappa<-2/3$. 
   \item[(iii)] Consider the family of McVittie spacetimes with $\Lambda=0$ and with 
   \be H(t) = \xi t^{-1},\quad \xi>\frac13,\quad  t>0. \ee 
   This corresponds to the background equation of state $P_0=\kappa\mu_0$ with $\kappa = -1+\frac{2}{3\xi}$. Then $\Sigma$ is timelike at late times for $\xi<2/3$ (or equivalently $\kappa>0$) and is spacelike at late times for $\xi>2/3$ (or equivalently $\kappa\in(-1,0)$).  
\end{itemize}
\end{proposition}

\noindent\textbf{Proof:}
With $\dot{t}=1$, we can write 
\be \Delta=\Delta^+ = -\gamma^{-1}(\gamma\dot{r}-\beta+1)(\gamma\dot{r}-\beta-1).\label{eq:delta-form} \ee
Using (\ref{eq:rdot-f}), we find 
\be \gamma\dot{r}-\beta = -b\left(1-\frac{r^3}{3m}\gamma^{1/2}H'\right) >0.\label{eq:grdot-minus-beta} \ee
Thus 
\be \sign{\Delta} = \sign{H'-Q} \label{eq:key-ineq} \ee
where
\be Q= \frac{3m}{r^3}\frac{(1-rJ)(1+\beta)}{\beta-rJ}f^{1/2},\label{eq:Q-def}
\ee

To prove part (i), we note that in the limit $t\to t_1^+$, we have $r\to2m$, $H\to H_1>0$, $J\to J_1>1/2m$ and $f\to0$. Thus $Q$ vanishes in this limit, but $H'$ has limiting value $H'(t_1)<0$. Then we must have $\Delta<0$ in this limit. This shows that there must exist some $t_3>t_1$ such that $\Sigma$ is spacelike for all $t\in(t_1,t_3)$. 

For the positive cosmological constant case, we can use the inequalities (\ref{a-props-pos}) and (\ref{b-props-pos}) of Lemma \ref{lemma:bounds} and a squeeze principle to obtain the asymptotic behaviours 
\be a \sim 1,\quad b \sim -\frac{m}{r^2H},\quad t\to+\infty. \label{a-b-limits} \ee
Then
\be \dot{r} \sim r (H_\infty - \frac13\frac{H'}{H_\infty}),\quad t\to+\infty, \label{rdot-lim-lam-pos} \ee
and integrating yields 
\be r \sim c\exp(H_\infty t),\quad t\to+\infty \label{r-lim-lampos}\ee
for some positive constant $c$. Note that the integral includes a term $\exp(H(t))$, which is asymptotically constant. With some manipulation (along the lines of those mentioned in the proof of Lemma \ref{lemma:bounds}: {multiply top and bottom by $\beta+rJ$ and use $f-r^2H^2=1-r^2J^2$}), we can write $Q$ as 
\begin{eqnarray} Q &=&  - \frac{(H+r^{-1}f^{1/2})(H+Jf^{1/2})}{J+r^{-1}}\frac{3f^{1/2}}{2r} \nonumber \\
&\sim& {-3\frac{H_\infty}{r},\quad t\to +\infty.} 
\end{eqnarray} 
Part (ii) follows by using (\ref{r-lim-lampos}) to compare this relation with 
\begin{eqnarray} H' &=& -\frac32H_\infty^2(1+k)\text{csch}^2 \left(\frac32(1+\kappa)H_\infty t\right) \nonumber \\
&\sim& -6H_\infty^2(1+\kappa)\exp\left(-3(1+\kappa)H_\infty t\right),\quad t\to+\infty, \label{eq:hp-limit} 
\end{eqnarray} 
and the result follows using (\ref{eq:key-ineq}). We note that the value of the positive constant $c$ is irrelevant for this comparison when $\kappa\neq-2/3$. 

To prove part (iii), we recall that $\dot{r}>0$ and that $r\to+\infty$ as $t\to+\infty$. Integrating the rough lower bound (\ref{eq:rdot-inf}) yields 
\be r \gtrsim ct^\xi,\quad t\to +\infty, \label{r-asymp-lamzero} \ee
and since $\xi>1/3$, it follows that $r^3H^2\to+\infty$ as $t\to+\infty$. Then 
\be J = (H^2+\frac{2m}{r^3})^{1/2} = H (1+o(1)),\quad t\to+\infty. \label{eq:J-H-remainder} \ee
Then the asymptotic relations (\ref{a-b-limits}) hold, and we find 
\be \frac{\dot{r}}{r} \sim \left(\xi+\frac13\right) t^{-1},\quad t\to+\infty, \label{rdot-lim-lamzero} \ee
and so 
\be r \sim c t^{(3\xi+1)/3},\quad t\to+\infty. \label{r-lim-lamzero} \ee
We can then determine the asymptotic behaviour of $Q$. We have the identity
\be Q = -\frac{3}{2}\frac{f^{3/2}}{r^2}\frac{(1+\beta)(\beta+rJ)}{1+rJ}.\label{Q-lamzero} \ee
From the calculations above, $\beta=rHf^{-1/2}\to0$, $rJ\to 0$ and $\beta+rJ\sim2rH$. Using (\ref{r-lim-lamzero}) and $H=\xi t^{-1}$, we obtain 
\be Q \sim -3\frac{\xi}{c}t^{-(\xi+4/3)}. \label{Q-lim-lamzero} \ee
Then we find that (\ref{eq:key-ineq}) is satisfied - and hence the hypersurface is timelike at large $t$ - whenever $\xi<2/3$ and is violated (spacelike at large $t$) whenever $\xi>2/3$. 
\hfill$\blacksquare$

\begin{comments}
    Notice that only the large$-t$ asymptotic behaviour of the Hubble function plays a role in the proofs of parts (ii) and (iii) here. It follows that we could weaken the  hypotheses to refer to McVittie spacetimes whose Hubble functions have the respective asymptotic behaviours. 
\end{comments}

\section{Conclusions}\label{sec:conclusions}

In the first instance, one has no right to expect that specifying the spacetime metric on one side of a matching hypersurface will allow one to determine both that hypersurface and the (entire) spacetime metric on the other side. The reason that this is in fact possible is that we have greatly restricted the unknown metric to be that of an FLRW spacetime. This is the context of our main result, Theorem \ref{Thm:Match-V4}. We can construct the unknown spacetime (and the matching hypersurface) because they depend only on functions of a parameter along the matching hypersurface. This is a highly specialized situation, that nevertheless answers one of our main questions: yes, a given McVittie spacetime can be smoothly matched to an FLRW spacetime. 

But does this provide a cosmological Oppenheimer-Snyder model? The answer here is no. What we mean by a cosmological OS model involves a spherical timelike matching hypersurface, so that the interior region is spatially bounded {- a star, let's call it -} providing the source of the exterior McVittie spacetime. Proposition \ref{prop:causal} shows that the matching hypersurface must \textit{always} admit a spacelike portion{: it cannot be the surface of a star}. We emphasize that this is an evolutionary feature of the matching hypersurface: if we were to assume that it is timelike at some suitable time $t>t_1$, then its past would include a spacelike portion (and its future may also do so: see parts (ii) and (iii) of the proposition). {The surface of the putative star would have been moving superluminally in the past (and may do so in the future).} {We emphasize the importance of using the general hypersurface formalism \cite{mars1993geometry} for this conclusion.} Had we assumed \textit{ab initio} that the matching hypersurface was timelike everywhere, and employed the classical Darmois junction conditions (continuity of the first and second fundamental forms of $\Sigma^\pm$), we could not have obtained a global existence result: a singularity of some form, presenting an obstacle to further evolution, would necessarily have appeared (presumably at a point at which $\Sigma$ becomes null in the general treatment). {The general hypersurface formalism accommodates the changing signature of the matching hypersurface, and allows us to draw the key conclusions above.}

Other necessary consequences of the situation being considered bear comment. First, we mention the fact that the embedded hypersurface must be degenerate in the sense of the failure of the maximal rank condition (\ref{max-rank}) at at least one point. However, our results show that this failure occurs at isolated points, and so does not appear to present a fundamental obstruction to constructing the matching configuration. (This perspective is motivated by the example of Figure 1, and is incorporated in Definition \ref{def:match}.) More concerning is the fact that the cosmic time coordinate $t$ of the McVittie region and that of the FLRW region ($\tau$) are not cosynchronous throughout $\Sigma$. As follows in part (iii) of Theorem \ref{Thm:Match-V4}, $d\tau/dt$ changes sign along $\Sigma$, being initially negative, and positive at late times. It follows that $\left.\tau\right|_{\Sigma}$ must have a minimum, and so matching \textit{does not} provide the entire FLRW metric. 

In conclusion, we have shown that there does not exist an isotropic source for the McVittie spacetime. It seems that the thin shell models of \cite{haines1993thin} and \cite{tang2025matching} with FLRW interior remain the best means of providing such a source. Some key questions remain unanswered in relation to these models, principally in relation to the global structure of the matching hypersurface. Does this meet the past singularity $\{r=2m\}$ at a finite time in the past, or does the shell successfully excise this from the spacetime? It would also be of interest to better understand the matter content of these shells: does the energy density remain positive and finite?

\section*{Acknowledgements}
I thank Nico Santos for discussions at an early stage of the development of this paper: the initial question of constructing an isotropic source for McVittie came from him (via Malcolm MacCallum, whom I also thank). Thanks also to Ra\"ul Vera for numerous discussions on matching, and to Eli\v{s}ka Kilme\v{s}ov\'{a} and Martin \v{Z}ofka for discussions on shell sources of McVittie spacetimes. {I am also very grateful to the referees whose comments have led to significant improvements in the paper.}

\bibliographystyle{unsrt}
\bibliography{mybib}

\end{document}